\documentclass[journal]{IEEEtran}
\usepackage{multicol}
\usepackage{amsmath,epsfig,comment}
\usepackage{amsthm}
\usepackage{tcolorbox}
\usepackage{amssymb}\usepackage{multirow}
\usepackage{mathrsfs}
\usepackage{amsmath}
\usepackage{arydshln}
\usepackage{multirow}
\usepackage{arydshln}

\usepackage{cases} 
\usepackage{amssymb,amsmath,cite}
\usepackage{epsfig}
\usepackage{color}
\usepackage{bm}
\usepackage{graphicx,subfigure}
\usepackage{algorithm}
\usepackage{algorithmic}
\usepackage{multirow}
\usepackage{soul}

\usepackage{extarrows}

\color{black}

\def\s{\mathbf{s}}

\def\ba{\mathbf{a}}
\def\bC{\mathbf{C}}
\def\bc{\mathbf{c}}
\def\bmu{\boldsymbol{\mu}}
\def\blam{\boldsymbol{\lambda}}
\def\I{\mathcal{I}}
\def\RR{\mathcal{R}}

\def\I{\mathcal{I}}
\def\R{\mathbb{R}}

\def\C{\mathbb{C}}
\def\bb{\mathbf{b}}
\def\bX{\mathbf{X}}
\def\bW{\mathbf{W}}
\def\bV{\mathbf{V}}
\def\bU{\mathbf{U}}
\def\bQ{\mathbf{Q}}

\def\H{\mathbf{H}}

\def\y{\mathbf{y}}

\def\x{\mathbf{x}}
\def\s{\mathbf{s}}

\def\h{\mathbf{h}}

\def\bA{\mathbf{A}}
\def\bM{\mathbf{M}}

\def\bD{\mathbf{D}}

\def\bH{\mathbf{H}}
\def\bB{\mathbf{B}}
\def\bG{\mathbf{G}}
\def\bh{\mathbf{h}}
\def\bthe{\boldsymbol{\Theta}}

\def\bw{\mathbf{w}}

\def\u{\mathbf{u}}

\def\bw{\mathbf{w}}
\def\bY{\mathbf{Y}}

\newtheorem{theorem}{Theorem}

\newtheorem{lemma}{Lemma}
\newtheorem{remark}{Remark}

\newtheoremstyle{noparens}%
  {}{}%
  {\itshape}{}%
  {\bfseries}{.}%
  { }%
  {\thmname{#1}\thmnumber{ #2}\mdseries\thmnote{ #3}}

\theoremstyle{noparens}

\title{Optimization of Beyond-Diagonal RIS: A Universal Framework Applicable to Arbitrary Architectures}
\author{\IEEEauthorblockN{Zheyu Wu and Bruno Clerckx,  \IEEEmembership{Fellow, IEEE}}
  	\thanks{Received 23 January 2026; revised 2 June 2026; accepted 29 June 2026. This work has been supported in part  by the UK Research and Innovation (UKRI) under Grant EP/Y004086/1, Grant EP/X040569/1, Grant EP/Y037197/1, Grant EP/X04047X/1, and Grant EP/Y037243/1. \emph{(Corresponding author: Bruno Clerckx.)}}

    	\thanks{Z. Wu is with the Department of Electrical and Electronic Engineering, Imperial College London, London, SW7 2AZ, U.K. (email: zheyu.wu@imperial.ac.uk). }
	
	\thanks{ B. Clerckx is with the Department of Electrical and Electronic Engineering, Imperial College London, London, SW7 2AZ, U.K, and also with Department of Electronic Engineering, Kyung Hee University, Yongin-si, Gyeonggi-do 17104, South Korea. (email: b.clerckx@imperial.ac.uk) }

  }
\date{today}
\begin{document}
\maketitle
\begin{abstract}
Reconfigurable intelligent surfaces (RISs) are envisioned as a promising technology for future wireless communication systems due to their ability to control the propagation environment in a hardware- and energy-efficient way. Recently, the concept of RISs has been extended to beyond-diagonal RISs (BD-RISs), which unlock the full potential of RISs thanks to the presence of tunable interconnections between RIS elements.
While various algorithms have been proposed for BD-RIS optimization, they mainly focus on specific  architectures whose scattering matrices exhibit very special structures. A universal optimization framework that can accommodate different BD-RIS circuit topologies is still lacking.
 In this paper, we  bridge this research gap by proposing an architecture-independent framework for BD-RIS optimization,  with the main focus on sum-rate maximization and transmit power minimization in multiuser multi-input single-output (MU-MISO) systems. 
  Specifically, we first incorporate BD-RIS architectures into the models  by connecting the scattering matrix with the admittance matrix and introducing appropriate constraints to the admittance matrix. The formulated problems are then solved by our custom-designed partially proximal alternating direction method of multipliers (pp-ADMM) algorithms. 
The pp-ADMM algorithms are computationally efficient, with each subproblem either admitting a closed-form solution or being easily solvable. 
 Simulation results demonstrate that the proposed approaches achieve a better trade-off between performance and computational efficiency compared to existing methods. 


\end{abstract}
\section{Introduction}
Reconfigurable intelligent surfaces (RISs) \cite{RIS1,RIS_sumrate,RIS_MM,RIS_tutorial,RIS_survey2} are envisioned as a promising technology to enhance the performance and coverage of future wireless communication systems. A RIS consists of numerous near-passive reconfigurable elements that can be coordinated to dynamically shape the propagation environment. Compared to traditional active devices like antennas or relays, RISs operate with significantly lower cost and power, making them a sustainable solution
 that meets the growing demand for low-cost, energy-efficient, and high-performance wireless networks in 6G and beyond. 

In a conventional RIS, each reconfigurable element is individually connected to its own load to ground without any inter-element connection, which results in a diagonal scattering matrix, also referred to as a phase shift matrix in the literature. Very recently, Beyond-diagonal RISs (BD-RISs) have emerged as an advanced generalization of conventional RISs  \cite{BDRIS,BDRIS_survey}. By allowing interconnections between different reconfigurable elements through tunable impedances, BD-RISs are able to achieve scattering matrices not limited to being diagonal, offering greater flexibility in beam control.  Different circuit topologies of inter-element connections result in different BD-RIS architectures.  
The fully-connected RIS, in which every pair of reconfigurable elements is inter-connected, offers the highest design flexibility but also entails the highest circuit complexity \cite{BDRIS}.  To mitigate the high complexity of fully-connected RISs, group-connected architectures have been introduced in \cite{BDRIS,group_conn}. Furthermore, two novel BD-RIS architectures, known as tree-connected and forest-connected RISs, have been developed in \cite{tree} and \cite{forest} using graph theory and have been shown to achieve the performance-complexity Pareto frontier in single-user multi-input single-output (MISO) systems.

A rich body of research has demonstrated the advantages of BD-RISs over conventional RISs in multiple aspects. BD-RIS is needed to achieve the upper bound on power/SNR maximization \cite{BDRIS} and is hence also needed to achieve any information theoretic bound of RIS-aided multiuser or multi-antenna systems. Beyond this, BD-RIS  offers numerous additional benefits, including boosting received signal power and sum-rate \cite{BDRIS,BDRIS_survey}, supporting various operational modes (reflective/transmissive/hybrid) \cite{group_conn}, expanding coverage \cite{coverage}, reducing the number of resolution bits \cite{discrete},  achieving better capability to overcome mutual coupling \cite{mutualcoupling,mutualcoupling2}, etc. The potential of BD-RISs has also been explored in TDD duplex systems (for channel attack) \cite{attack},   wideband communication systems \cite{OFDM,OFDM2,OFDM3}, full-duplex systems \cite{duplex}, dual-function radar-communication (DFRC) systems \cite{DFRC}, unmanned aerial vehicle (UAV) systems \cite{UAV}, and rate-splitting multiple access (RSMA) systems \cite{RSMA0,RSMA}.  Furthermore, several extensions of BD-RIS have been proposed to enhance its performance, such as new grouping strategies adaptive to channel state information (CSI) \cite{grouping,grouping2}, a distributed BD-RIS deployment that is less susceptible to lossy interconnections \cite{lossy}, and a BD-RIS aided stacked intelligent metasurface with improved  wave manipulation capabilities \cite{SIM}.  

Despite the great benefits offered by BD-RISs, their optimization is far more challenging than that of conventional RISs. For conventional RISs, the scattering matrix is diagonal, with a unit-modulus constraint on each diagonal entry. Various optimization methods have been developed to handle this constraint, including semidefinite relaxation (SDR) \cite{RIS1},   manifold optimization \cite{RIS_sumrate}, and majorization-minimization \cite{RIS_MM}.    These algorithms are tailored to the diagonal phase-shift structure and cannot be directly applied to BD-RISs. The more complicated circuit topology of BD-RISs introduces more complicated objective functions and constraints to the corresponding optimization problems, such as unitary and/or symmetric constraints \cite{BDRIS,group_conn,coverage,grouping,grouping2,duplex,channel,RSMA0,RSMA,OFDM2,OFDM3,UAV,DFRC} and constraints/objective functions involving matrix inversion \cite{mutualcoupling,mutualcoupling2,discrete,tree,forest,lossy,OFDM}. 
In addition, compared with conventional RIS, the number of optimization variables increases dramatically due to the significantly larger number of tunable impedances in BD-RISs. For example, in fully-connected BD-RISs, the number of impedances grows quadratically with the number of RIS elements, which imposes higher demands on algorithmic efficiency.

Existing optimization strategies for BD-RISs can mainly be catagorized into  closed-form approaches, heuristic approaches, and optimization-based approaches. Closed-form approaches seek the closed-form solution of the corresponding optimization problem \cite{closeform, tree,forest,mutualcoupling2}. However, this  kind of method is applicable only in very limited scenarios, such as single-user single-input single-output (SISO) systems.  Heuristic approaches \cite{twostage} first determine the scattering matrix heuristically and then treat it as fixed in beamforming design. Although simple, such approaches lack theoretical support and do not offer performance guarantees. Optimization-based approaches leverage classical optimization techniques to tackle the complexities  imposed by BD-RISs. Examples include employing manifold optimization for unitary constraints \cite{group_conn,grouping,yang}, applying penalty dual decomposition to handle unitary and symmetric constraints \cite{PDD},  applying Neumann series approximation \cite{mutualcoupling} and quasi-Newton method \cite{RSMA0} to deal with matrix inversion, etc. These approaches generally achieve satisfactory performance, but at the expense of higher computational cost than heuristic approaches. The key to designing optimization-based approaches is to strike a balance between efficiency and effectiveness.

Although various approaches have been proposed for BD-RIS optimization in different scenarios, existing works have the following limitations. First and most importantly, {\color{black} all existing models  for BD-RIS optimization in multiuser systems are formulated based on the scattering matrix. However, there is generally no clear relationship between the scattering matrix and the inherent circuit topology of BD-RISs. Consequently, existing models are fundamentally limited and can only handle very special BD-RIS architectures, e.g., fully-connected and group-connected BD-RIS \cite{group_conn,PDD,twostage} (Please see  Remark \ref{remark1} for more detailed discussions).}
To date, an architecture-independent optimization framework that can accommodate different BD-RIS circuit topologies in a unified manner is still lacking for general multiuser systems.
 Second, existing algorithms are  either heuristic \cite{twostage} or time-consuming \cite{PDD}. In particular, very few works deal with the transmit power minimization problem, possibly due to the technical challenges arising from the coupling of  BD-RIS constraints with users' quality-of-service (QoS) constraints \cite{PDD}. To the best of our knowledge, the only existing algorithm is \cite{PDD}, which, however, is computationally inefficient when the number of RIS elements is large.

Motivated by the aforementioned limitations, this paper aims to develop an efficient architecture-independent framework for BD-RIS  optimization,  considering both sum-rate maximization and transmit power minimization problems. The contributions are summarized as follows. 

 \emph{{\color{black}1)  New optimization models that accommodate different BD-RIS architectures. }}
We formulate the sum-rate maximization and transmit power minimization problems for BD-RIS-aided multiuser MISO (MU-MISO) systems in an architecture-independent way such that the models can accommodate different BD-RIS architectures. Specifically, building on microwave network theory and existing BD-RIS modeling \cite{tree,generalmodel}, we relate the scattering matrix with the admittance matrix and incorporate BD-RIS architecture into the model by imposing proper sparsity constraints on the admittance matrix. 
The admittance matrix directly relates to the circuit topology of the reconfigurable impedance network and provides a sparser and more flexible mathematical representation than the scattering matrix. For example, for group-connected BD-RISs, both the scattering and admittance matrices are block diagonal. However, for more efficient architectures such as tree- and forest-connected BD-RISs \cite{tree,forest}, the admittance matrices are sparse, whereas the corresponding scattering matrices are  full. Therefore, using the admittance matrix as the optimization variable is essential for explicitly capturing different BD-RIS architectures, whereas the scattering matrix generally does not directly reveal the underlying circuit topology. 

{\color{black}\emph{2) Efficient algorithms for solving the formulated problems.} Compared with existing BD-RIS optimization models, the formulated problems are more challenging to solve due to  the complicated non-convex constraint that relates the scattering and admittance matrices and the additional architecture-specific constraint on the admittance matrix. 
We propose efficient partially proximal alternating direction method of multipliers (pp-ADMM) algorithms for solving them. }
 The key idea is to first introduce appropriate auxiliary variables to {\color{black}greatly reduce the problem dimension} and make the problems amenable to the ADMM framework.
 This includes transforming the constraint involving matrix inversion into a low-dimensional bilinear constraint, simplifying the sum-rate maximization problem with fractional programming (FP) technique, and decomposing the complicated QoS constraint in the transmit power minimization problem into a bilinear constraint and an easy-to-projection constraint. {\color{black}Then, a pp-ADMM algorithm is designed to solve the problem, where appropriate proximal terms are introduced into the updates of specific variables to enhance stability and guarantee convergence. The proposed algorithm is computationally efficient and enjoys desirable theoretical convergence properties.}


\emph{3) Generalization of the proposed approach}. We extract an optimization framework to handle general utility functions and, in particular, explore its applicability to max-min fairness and energy efficiency maximization problems. In addition, we discuss how the proposed framework can be extended to more general multiuser multiple-input multiple-output (MU-MIMO) systems, statistical channel state information (CSI) settings, practical finite-resolution BD-RIS implementations, and other BD-RIS operational modes.

To the best of our knowledge, the proposed approach is the first architecture-independent framework for  BD-RIS optimization { in general multiuser systems}. In addition, simulation results demonstrate that, when applied to group- and fully-connected BD-RISs, our approach achieves a superior trade-off between sum-rate performance and computational efficiency compared to the existing state-of-the-art algorithms in \cite{PDD,twostage}. For the transmit power minimization problem, our approach attains lower transmit power with significantly reduced computational complexity than the only existing algorithm in \cite{PDD}. 

\emph{Organization}: The rest of the paper is organized as follows.  We first introduce the system model and problem formulations in Section \ref{sec:2}. In Sections \ref{sec:3} and \ref{sec:4}, we propose efficient pp-ADMM algorithms for solving the sum-rate maximization and transmit power minimization problems of BD-RIS-aided MU-MISO systems, respectively.  In Section \ref{sec:generalization}, we discuss the generalization of the proposed optimization framework. 
The effectiveness and efficiency of the proposed approaches are verified through extensive simulations in Section \ref{sec:5}. Finally, the paper is concluded in Section \ref{sec:6}. 

\emph{Notations}: Throughout the paper, $\mathbb{C}$ and $\mathbb{R}$ denote the complex space and the real space, respectively. We use $x$, $\x$, $\mathbf{X}$, and $\mathcal{X}$ to represent a scalar, column vector, matrix, and set, respectively. For a matrix $\mathbf{X}$,  $\mathbf{X}_{\mathcal{S}_1,\mathcal{S}_2}$ denotes its submatrix with  rows indexed by $\mathcal{S}_1$ and columns indexed by $\mathcal{S}_2$, where  $\mathbf{X}(\mathcal{S}_1,\mathcal{S}_2)$ is also used in certain context to avoid possible confusion. In particular, $X_{n,m}$ denotes the $(n,m)$-th entry of $\mathbf{X}$. The operators $(\cdot)^T$, $(\cdot)^\dagger$, $(\cdot)^{-1}$, $\RR(\cdot)$, and $\I(\cdot)$ return the transpose, the Hermitian transpose, the inverse, the real part, and the imaginary part of their corresponding argument, respectively.     Given two vectors $\x\in\C^{n\times 1}$ and $\y\in\C^{n\times 1}$, $\left<\x,\y\right>$ represents their inner product, i.e., $\left<\x,\y\right>=\RR(\x^\dagger\y)$. For two sets $\mathcal{X}_1$ and $\mathcal{X}_2$, $\mathcal{X}_1\backslash \mathcal{X}_2$ denotes the difference between $\mathcal{X}_1$ and $\mathcal{X}_2$. 
  The notation  $\|\cdot\|_2$  refers to the  $\ell_2$-norm of a vector or the spectral norm of a matrix, and $\|\cdot\|_F$ denotes the Frobenius norm of a matrix. The operator $|\cdot|$ returns the modulus if applied to a scalar and the number of elements if applied to a set. The symbol $\mathbf{I}_n$  refers to an $n$-dimensional identity matrix. Finally, $\mathsf{i}$ represents the imaginary unit.



\section{Problem Formulation}\label{sec:2}
\subsection{System Model}
Consider a BD-RIS-aided MU-MISO system\footnote{We discuss the generalization of the proposed approach to the MU-MIMO system in Section \ref{sec:generalization}.}, where a base station (BS) with $N$ transmit antennas serves $K$ single-antenna users simultaneously with the assistance of an $M$-element BD-RIS. As in \cite{BDRIS}, the $M$-element BD-RIS can be modeled as $M$ antennas connected to an $M$-port reconfigurable impedance network, which induces a non-diagonal scattering matrix $\bthe\in\C^{M\times M}$ to control the propagation of the incident waveform.   Denote $\s=[s_1,s_2,\dots, s_K]^T$ as the data symbol vector for the users with $\mathbb{E}\left[\s\s^\dagger\right]=\mathbf{I}_K$, and let $\bW=[\bw_1,\bw_2,\dots,\bw_K]\in\C^{N\times K}$ be the beamforming matrix. 
Then the received signal at the $k$-th user can be expressed as\footnote{In this paper, we assume perfect matching, no mutual coupling,  no structural scattering or specular reflection, and that the unilateral approximation holds \cite{generalmodel}. } 
\begin{equation}\label{sys_model}
r_k=\bh_k^\dagger\bthe\bG\sum_{k=1}^K\bw_ks_k+n_k,~k=1,2,\dots, K, 
\end{equation}
where $\bG\in\C^{M\times N}$ is  the channel matrix between the BS and the BD-RIS, $\bh_k\in\C^{M\times 1}$ denotes  the channel vector  between the BD-RIS and the $k$-th user, and  $n_k\sim\mathcal{CN}(0,\sigma^2)$ is the additive white Gaussian noise. Throughout this paper, we assume that perfect instantaneous  channel state information (CSI) is available at the BS.  For simplicity, we also neglect the direct channel between the BS and the users; however, all results can be straightforwardly generalized to scenarios where the direct channel is taken into account.

Our goal in this paper is to develop a framework for jointly optimizing the beamforming matrix $\bW$ and the scattering matrix $\bthe$, in a universal manner such that the framework is architecture-independent, i.e., it is  applicable to any given BD-RIS architecture; see Section \ref{sec:RIS_arch} for detailed discussions on BD-RIS architectures. 
 For this purpose, we resort to the microwave network theory \cite{microwavebook} and relates the scattering matrix $\bthe$ (i.e., the $S$-parameter) with the admittance matrix $\mathbf{Y}\in\C^{M\times M}$ (i.e., the $Y$-parameter) as 
\begin{equation}\label{theta_Y}
\bthe=(\mathbf{I}+Z_0\mathbf{Y})^{-1}(\mathbf{I}-Z_0\mathbf{Y}),
\end{equation}
where $Z_0$ denotes the reference impedance, usually set as $Z_0=50\,\Omega$.
 The  $Y$-parameter  has the advantage of being able to characterize the circuit topology of BD-RISs. Specifically,  denote the admittance for the $n$-th port of the reconfigurable impedance network as $\bar{Y}_n$,  and denote the admittance connecting the $n$-th and $m$-th ports as $\bar{Y}_{n,m}$. According to \cite{tree}, the admittance matrix $\mathbf{Y}$ of the reconfigurable impedance network  is given by
\begin{equation}\label{admittance}
Y_{n,m}=\left\{
\begin{aligned}
&-\bar{Y}_{n,m},~~~~~~~~~\,~~&\text{if}~n\neq m;\\
&\bar{Y_n}+\textstyle\sum_{j\neq n}\bar{Y}_{n,j},~~&\text{if }n=m.
\end{aligned}\right.
\end{equation}
Clearly, for $n\neq m$, $Y_{n,m}$ is nonzero if and only if $\bar{Y}_{n,m}$ is nonzero. Therefore,  the interconnections between different ports in the reconfigurable impedance network are represented by the non-zero elements in $\mathbf{Y}$. 

In the following, we review several popular   architectures proposed in the existing literature (see also Fig. \ref{topology}). 
\subsection{Different RIS Architectures}\label{sec:RIS_arch}
\subsubsection{Single-Connected RIS} When there is no interconnection between ports, i.e., $Y_{n,m}=0$ for all $n\neq m$,  the BD-RIS reduces to the traditional diagonal RIS, also known as a single-connected RIS, which has been extensively studied in literature. In this case, $\bY$ and $\bthe$ are both diagonal.
\begin{figure*}
\centering
\includegraphics[scale=0.25]{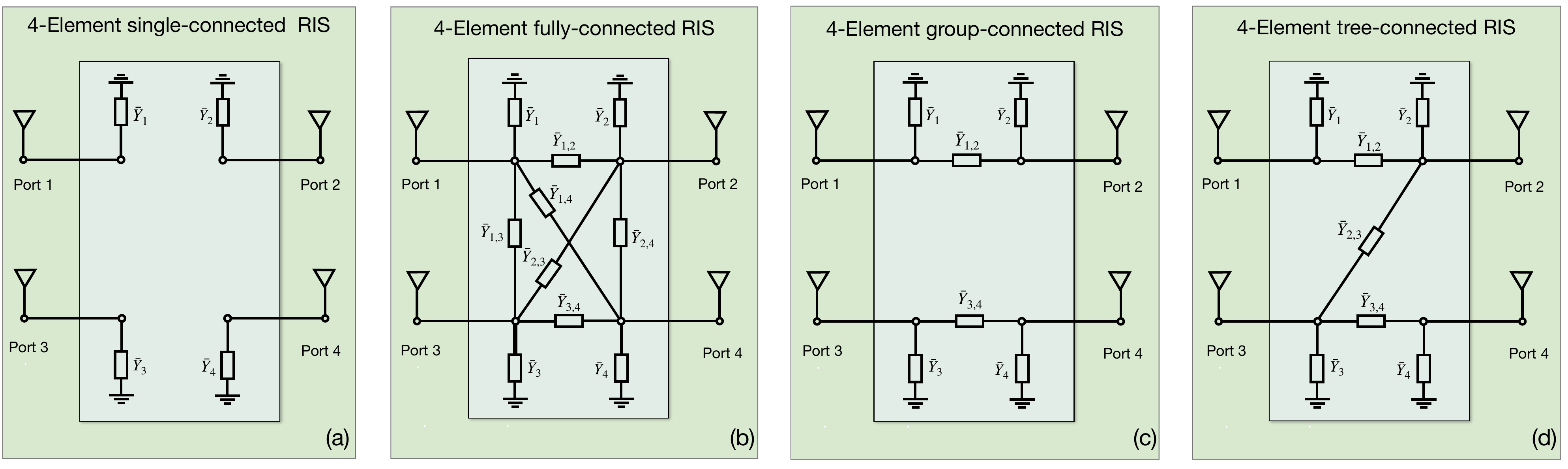}
\caption{Different BD-RIS architectures, where (a) single-connected RIS, (b) fully-connected RIS, (c) group-connected RIS, (d) tree-connected RIS. }
\label{topology}
\end{figure*}
\subsubsection{Fully-Connected RIS} When there is an impedance connecting each pair of ports, the corresponding architecture is referred to as a fully-connected RIS. In this case, both $\bY$ and $\bthe$ are full matrices, i.e., $Y_{n,m}\neq 0$ and $\Theta_{n,m}\neq 0$ for all $n,m\in\{1,2,\dots, M\}$. The fully-connected RIS offers the highest design flexibility but also suffers from the highest circuit complexity.
\subsubsection{Group-Connected RIS}\label{sec:group} The group-connected RIS has been proposed to achieve a favorable trade-off between system performance and circuit complexity \cite{group_conn}. In group-connected RIS, the RIS elements are equally divided into multiple groups, where each group employs a fully-connected architecture, and no connections exist between different groups.
 Let $G$ denote the number of groups; then each group has a number of $G_s=M/G$ RIS elements, which we refer to as the group size. In this case, both $\bY$ and $\bthe$ exhibit  a block diagonal structure. 
In particular, a group-connected RIS becomes a fully-connected RIS when $G=1$ and reduces to a single-connected RIS when $G=M$.
\subsubsection{Tree-Connected RIS}
In \cite{tree}, the circuit topology of  BD-RIS  is characterized by a graph $\mathcal{G}=(\mathcal{V},\mathcal{E})$, where the vertex set $\mathcal{V}:=\{V_1,V_2,\dots, V_M\}$ represents the ports of the reconfigurable impedance network,  and the edge set $\mathcal{E}:=\left\{(V_i,V_j)\mid Y_{i,j}\neq 0,~i\neq j\right\}$ records all the interconnections between  ports. Tree-connected RISs refer to a class of BD-RIS architectures whose corresponding graph is a tree \cite{tree}. One example of tree-connected RISs is the  tridiagonal RIS \cite{tree}, whose circuit topology forms a path graph, and the corresponding admittance matrix is a tridiagonal matrix given by
\begin{equation}\label{tri}\bY=\left[\begin{matrix}
Y_{1,1}&Y_{1,2}&&\vspace{-0.1cm}\\
Y_{1,2}&Y_{2,2}&\ddots&\vspace{-0.0cm}\\
&\hspace{-0.5cm}\ddots&\hspace{-0.2cm}\ddots&Y_{M-1,M}\vspace{0.1cm}\\
&&\hspace{-0.5cm}Y_{M-1,M}&\hspace{-0.3cm}Y_{M,M}
\end{matrix} \right].\end{equation}
Due to the nonlinear relationship between $\bthe$ and $\bY$ in \eqref{theta_Y}, the scattering matrix $\bthe$ of a tree-connected RIS does not exhibit a clear structure and can, in general, be a full matrix.

To illustrate this point, the following remark provides an example and further explains why it is necessary to work with the admittance matrix $\bY$ to characterize arbitrary BD-RIS architectures.
\begin{remark}\label{remark1}
The one-to-one mapping between non-zero entries in  $\bY$  and physical interconnections between ports does not hold for the scattering matrix $\bthe$, i.e., a port interconnection does not necessarily correspond to a non-zero entry in $\bthe$, and vice versa. This is because the mapping from $\bY$ to $\bthe$, given by \eqref{theta_Y}, is inherently nonlinear. As a result, the scattering matrix can only be used to characterize BD-RIS architectures that 
 exhibit very special structures, e.g., group-connected RISs, where a block-diagonal $\bY$ leads to a block-diagonal $\bthe$, as discussed in Section \ref{sec:group}. For general BD-RIS architectures, including tree-connected RIS and the more recent architectures proposed in \cite{graphtit}, the corresponding scattering matrix $\bthe$ does not exhibit any exploitable structure. 

To better illustrate this, we take a $4$-element tridiagonal RIS as an example, where the $i$-th port is connected to the $(i+1)$-th port with admittance $\bar{Y}_{i,i+1}=\frac{1}{50}$ siemens (S),~$i=1,2,3$, and each port has a self-admittance $\bar{Y}_{i}=\frac{1}{50}$ S,~$i=1,2,3,4$. According to \eqref{admittance}, the  admittance matrix 
$\bY$ is$$\bY=\frac{1}{50}\left[\begin{matrix}2&-1&0&0\\-1&2&-1&0\\0&-1&2&-1\\0&0&-1&2\end{matrix}\right],$$ which has a tridiagonal structure in \eqref{tri}, while the scattering matrix, according to \eqref{theta_Y}, is 
$$\bthe\hspace{-0.05cm}=\hspace{-0.05cm}\frac{1}{10}\hspace{-0.1cm}\left[\begin{matrix}-4.0-7.5\mathrm{i}\hspace{-0.25cm}&4.6-1.0\mathrm{i}\hspace{-0.25cm}&2.1+1.0\mathrm{i}\hspace{-0.25cm}&0.7+0.8\mathrm{i}\\4.6-1.0\mathrm{i}\hspace{-0.25cm}&-0.9-6.6\mathrm{i}&5.2-0.2\mathrm{i}\hspace{-0.25cm}&2.1+1.0\mathrm{i}\hspace{-0.25cm}\\2.1+1.0\mathrm{i}\hspace{-0.25cm}&5.2-0.2\mathrm{i}\hspace{-0.25cm}& -1.9 - 6.5\mathrm{i}\hspace{-0.25cm}& 4.6 - 1.0\mathrm{i}\\0.7+0.8\mathrm{i}\hspace{-0.25cm}&2.1+1.0\mathrm{i}\hspace{-0.2cm}&4.6 - 1.0\mathrm{i}\hspace{-0.25cm}&-4.0 - 7.5\mathrm{i}\hspace{-0.05cm}\end{matrix}\right]\hspace{-0.05cm},$$ 
which is a full matrix without any exploitable structure. \end{remark}

{To maximize the reflected power of BD-RISs, the admittance matrix $\mathbf{Y}$ of the reconfigurable impedance network should be a purely imaginary matrix,  denoted as  $\bY=\mathsf{i}\bB$, where $\bB\in\R^{M\times  M}$ is the susceptance matrix. 
    Equivalently, this amounts to a unitary constraint on   $\bthe$, i.e., $\bthe^{\dagger}\bthe=\mathbf{I}_M$. As commonly assumed in the literature,  the reconfigurable impedance network is considered reciprocal, leading to symmetric constraints:
$\bthe=\bthe^T$,  $\bY=\bY^T$, and $\bB=\bB^T$.  }

It has been proved in \cite{tree} that the tree-connected RIS is optimal for single-user MISO systems under the lossless and reciprocal setup, which achieves the same channel gain as the fully-connected RIS while maintaining the lowest circuit complexity. However, its performance in multiuser systems remains unknown due to the lack of a universal optimization framework applicable to arbitrary architectures. Throughout this paper, an arbitrary BD-RIS architecture refers to any given circuit topology of the  lossless and reciprocal reconfigurable impedance network, which is encoded through the sparsity pattern of the admittance/susceptance matrix; see \eqref{arch_constraint}.

 \subsection{Problem Formulation}
 With the system model in \eqref{sys_model}, the signal-to-interference-plus-noise ratio (SINR) of the $k$-th user is 
 \begin{equation}\label{SINR}
 \text{SINR}_k(\bthe,\bW)=\frac{|\bh_k^\dagger\bthe\bG\bw_k|^2}{\sum_{j\neq k}|\bh_k^\dagger\bthe\bG\bw_j|^2+\sigma^2},~k=1,2,\dots, K.
 \end{equation}
To account for different BD-RIS architectures, we introduce the following  notation: 
\begin{equation}\label{arch_constraint}
 \begin{aligned}\mathcal{B}=\{\mathbf{B}\in\R^{M\times M}\mid& B_{i,j}=0,~\text{if there is no interconnection}\\
&\text{between the $i$-th and $j$-th ports},~i\neq j\},\end{aligned}
\end{equation}
i.e., the set $\mathcal{B}$ collects all susceptance matrices corresponding to a specified BD-RIS circuit topology.

 In this paper, we mainly focus on  two different problem formulations:  sum-rate maximization under the total transmit power constraint and transmit power minimization under users' QoS constraints. The proposed framework, however, can  be generalized to accommodate other utility functions; see discussions in  Section \ref{sec:generalization}. 

\subsubsection{Sum-rate Maximization} Based on the above discussions, the sum-rate maximization problem can be formulated as 
\begin{subequations}\label{multiuser}
\begin{align}
\max_{\bW,\bthe,\bB}~&\sum_{k=1}^K\log\left(1+\text{SINR}_k(\bthe,\bW)\right)\\
\text{s.t.}~~ &\bthe=\left(\mathbf{I}+\mathrm{i}Z_0\mathbf{B}\right)^{-1}\left(\mathbf{I}-\mathrm{i}Z_0\mathbf{B}\right),~~\label{con:theta}\\
~~&\bB=\bB^{T},~\bB\in\mathcal{B},\label{con:B}\\
&\|\bW\|_F^2\leq P_T,\label{con:power}
\end{align}
\end{subequations}
where constraint \eqref{con:theta} relates the scattering matrix $\bthe$ with $\bB$, constraint  \eqref{con:B} is to ensure that the reconfigurable impedance network is reciprocal\footnote{For non-reciprocal RISs, there is no symmetric constraint, which simplifies the corresponding optimization problem. Our proposed framework is directly applicable in this case.} and complies with the specified architecture, and \eqref{con:power} is the transmit power constraint with $P_T$ being the maximum transmit power. 
\subsubsection{Transmit Power Minimization} The transmit power minimization problem is formulated as 
\begin{subequations}\label{min_power}
\begin{align}
\min_{\bW,\bthe,\bB}~&\|\bW\|_F^2\\
\text{s.t. }~~&\text{SINR}_k(\bthe,\bW)\geq \Gamma_k,~k=1,2,\dots, K,\label{con:QoS}\\
&\bthe=(\mathbf{I}+\mathrm{i}Z_0\bB)^{-1}(\mathbf{I}-\mathrm{i}Z_0\bB),\label{con:theta2}\\
&\bB=\bB^T,~\bB\in\mathcal{B},\label{con:B2}
\end{align}
\end{subequations}
where  $\Gamma_k>0$ is the pre-specified SINR threshold for user $k$.
 
  Both the sum-rate maximization and  transmit power minimization problems have been extensively studied in traditional MU-MISO systems \cite{sumrate2,WMMSE,SOCP,duality} and diagonal RIS-aided MU-MISO systems \cite{RIS1,RIS_sumrate,RIS_MM,RIS_tutorial,RIS_survey2}. {\color{black}Recent studies have also developed optimization algorithms for BD-RIS-aided systems \cite{group_conn,grouping,yang,PDD,twostage}.  Compared with these works, our formulated models are considerably more universal, as they accommodate arbitrary BD-RIS architectures, whereas prior works \cite{group_conn,grouping,yang,PDD,twostage} cannot because they are directly modeled using the scattering matrix (see the discussion in Remark \ref{remark1}).}
  
  {\color{black}  However, the new problem formulations also introduce new challenges for algorithm design. Specifically, incorporating arbitrary BD-RIS architectures requires formulating the problems with the admittance (susceptance) matrix. This results in constraints involving matrix inversion, as shown in \eqref{con:theta} and \eqref{con:theta2}, which is further coupled with the architecture-specific constraint \eqref{con:B} and \eqref{con:B2}, making our problems significantly more complex than existing BD-RIS optimization models.}

   In the following two sections, we will propose two efficient algorithms for solving problems \eqref{multiuser} and \eqref{min_power} by carefully exploiting the structure of the problems.

 \section{Sum-Rate Maximization Problem}\label{sec:3}
In this section, we propose an efficient algorithm for solving the sum-rate maximization problem in \eqref{multiuser}. First, we transform \eqref{multiuser} into a more amenable form in Section \ref{sec:algorithm1}, which is then solved by a custom-designed partially proximal ADMM (pp-ADMM) algorithm in Section \ref{sec:algorithm2}. The complexity and convergence analysis of the proposed pp-ADMM algorithm is given in Section {\ref{sec:algorithm3b}.}
\subsection{An Equivalent Formulation}\label{sec:algorithm1}
{\color{black} An important observation we can draw from the SINR expression in \eqref{SINR} is that $\text{SINR}_k$ depends on $\bthe$ only through the  vector $\bthe^{\dagger}\h_k\in\C^{M\times 1}$.} {\color{black} Since the number of users is typically much smaller than the number of RIS elements, this motivates us to replace $\bthe$ with a low-dimensional variable $$\bU=\bthe^\dagger \bH\in\C^{M\times K},$$ where $\bH=[\h_1,\h_2,\dots,\h_K]$ and  $\bU=[\u_1,\u_2,\dots, \u_K]$.} With this new variable, the SINR of user $k$ can be expressed as 
a function of $\u_k$ and $\bW$ as 
\begin{equation}\label{SINR:ukW}
\text{SINR}_k(\u_k,\bW)=\frac{|\u_k^\dagger\bG\bw_k|^2}{\sum_{j\neq k}|\u_k^\dagger\bG\bw_j|^2+\sigma^2}.
\end{equation}
{\color{black}In addition, rewriting \eqref{con:theta} (and \eqref{con:theta2}) in terms of variable $\bU$ yields
$$\left(\mathbf{I}-\mathrm{i}Z_0\mathbf{B}\right)\bU=\left(\mathbf{I}+\mathrm{i}Z_0\mathbf{B}\right)\bH,$$
which reduces the dimension of the constraint from $M^2$ to $MK$.}

With the new variable $\bU$, problem \eqref{multiuser} is now transformed into the following equivalent form:
 \begin{subequations}\label{eqn:reformulate1}
\begin{align}
\max_{\bW,\bB,\bU}~&R(\bW,\bU):=\sum_{k=1}^K\log\left(1+\text{SINR}_k(\u_k,\bW)\right)\label{obj}\\
\text{s.t.}~\,~~ &\left(\mathbf{I}-\mathrm{i}Z_0\mathbf{B}\right)\bU=\left(\mathbf{I}+\mathrm{i}Z_0\mathbf{B}\right)\bH,\label{con:bilinear}\\
~~&\bB=\bB^{T},~\bB\in\mathcal{B}, \label{con:symmetric2} \\
&\|\bW\|_F^2\leq P_T.\label{con:power2}
\end{align}
\end{subequations}
{\color{black} Compared with \eqref{multiuser}}, the advantage of this new formulation is twofold. First,  there is no longer any matrix inversion in the constraints. The complicated constraint involving matrix inversion in \eqref{con:theta} is transformed into the bilinear constraint in \eqref{con:bilinear}, which is more trackable. Second, it greatly reduces the dimension  of both the optimization variables and constraints. {\color{black}By working with the lower-dimensional problem in \eqref{eqn:reformulate1}, the per-iteration computational complexity of the proposed approach is upper bounded by \( \mathcal{O}(M^3 K^3) \) (as will be discussed in Section~\ref{sec:algorithm3b}), whereas directly handling the original formulation would result in a complexity of $\mathcal{O}(M^6)$.}

Since the objective function \eqref{obj} is still complicated, we employ the FP technique \cite{FP,FP2} to rewrite it into the following equivalent form:
\begin{equation}\label{FP:sumrate}
R(\bW,\bU)=\max_{\y,\boldsymbol{\gamma}}\tilde{R}(\y,\boldsymbol{\gamma},\bW,\bU),
\end{equation}
where 
\begin{equation*}\label{def:F}
\begin{aligned}
&\tilde{R}(\y,\boldsymbol{\gamma},\bW,\bU)=\sum_{k=1}^K \log(1+\gamma_k)-\gamma_k\\
&+(1+\gamma_k)\bigg(2\left<y_k,\u_k^\dagger\bG\bw_k\right>-|y_k|^2\big(\sum_{j=1}^K|\u_k^\dagger\bG\bw_j|^2+\sigma^2\big)\bigg).
\end{aligned}\end{equation*}
With the above transformation, problem \eqref{eqn:reformulate1}  can be equivalently expressed as 
    \begin{equation}\label{eqn:reformulate2}
\begin{aligned}
\max_{\y, \boldsymbol{\gamma},\bW,\bB,\bU}~&\tilde{R}(\y,\boldsymbol{\gamma},\bW,\bU)\\\
\text{s.t.}~~~~~~&\eqref{con:bilinear}-\eqref{con:power2}.
\end{aligned}
\end{equation}
Problems \eqref{eqn:reformulate1} and \eqref{eqn:reformulate2} are equivalent in the sense that given any optimal solution $(\bW^*,\bB^*,\bU^*)$ of problem \eqref{eqn:reformulate1}, $(\y^*,\boldsymbol{\gamma}^*,\bW^*,\bB^*,\bU^*)$ with
$$y_k^*=\frac{{\u_k^*}^\dagger\bG\bw^*_k}{\sum_{j=1}^K|{\u^*_k}^\dagger\bG\bw^*_j|^2+\sigma^2}$$ and $$\gamma_k^*=\frac{|{\u^*_k}^\dagger\bG\bw^*_k|^2}{\sum_{j\neq k}|{\u^*_k}^\dagger\bG\bw^*_j|^2+\sigma^2}$$
is an optimal solution  of problem \eqref{eqn:reformulate2}, and vise versa \cite{FP2}.
\subsection{pp-ADMM Algorithm for Solving \eqref{eqn:reformulate2}}\label{sec:algorithm2}
{Compared to the original problem \eqref{multiuser}, problem \eqref{eqn:reformulate2} is much easier to tackle. Its objective function $\tilde{R}(\y,\boldsymbol{\gamma},\bW,\bU)$ is concave and straightforward to  optimize with respect to each variable block when the other blocks are fixed. In addition, constraints \eqref{con:symmetric2} and \eqref{con:power2} are simple constraints on $\bB$ and $\bW$, respectively, and do not include coupling among different variable blocks.   The only remaining challenge lies in the bilinear constraint \eqref{con:bilinear} that couples variables $\bB$ and $\bU$. This structure naturally motivates the use of an ADMM-type framework \cite{boyd2011distributed,survey}. By incorporating the bilinear equality constraint into the augmented Lagrangian, the coupling between $\mathbf{B}$ and $\mathbf{U}$ can be handled through a quadratic penalty term.  Consequently, the resulting subproblems with respect to different variable blocks are tractable and either admit closed-form solutions or can be solved efficiently, which will be detailed later.}

In the following, we propose a partially proximal ADMM (pp-ADMM) algorithm for solving \eqref{eqn:reformulate2}. First, the Lagrangian function of \eqref{eqn:reformulate2} is given by
   \begin{equation}\label{lagrangian}
    \begin{aligned}
&\mathcal{L}_{\rho}(\y,\boldsymbol{\gamma},\bW,\bB,\bU,\blam)\\
=\,&\tilde{R}(\y,\boldsymbol{\gamma},\bW,\bU)-\left<\blam,(\mathbf{I}-\mathrm{i}Z_0\bB)\bU-(\mathbf{I}+\mathrm{i}Z_0\bB)\bH\right>\\
&-\frac{\rho}{2}\left\|(\mathbf{I}-\mathrm{i}Z_0\bB)\bU-(\mathbf{I}+\mathrm{i}Z_0\bB)\bH\right\|_F^2,
\end{aligned}
\end{equation}
where $\blam\in\C^{M\times K}$ is the Lagrange multiplier and $\rho$ is the penalty parameter. The proposed pp-ADMM algorithm is given as follows: 
\begin{subequations}\label{ppADMM}
\begin{align}
\y^{t+1}&\in\arg\max~\mathcal{L}_{\rho}(\y,\boldsymbol{\gamma}^t,\bW^t,\bB^t,\bU^{t},\blam^t),\label{suby}\\
\boldsymbol{\gamma}^{t+1}&\in\arg\max~\mathcal{L}_{\rho}(\y^{t+1},\boldsymbol{\gamma},\bW^t,\bB^t,\bU^{t},\blam^t),\label{subgamma}\\
\bW^{t+1}&\in\arg\max_{{\scriptsize\|\bW\|_F^2\leq P_T}}\mathcal{L}_{\rho}(\y^{t+1},\boldsymbol{\gamma}^{t+1},\bW,\bB^t,\bU^{t},\blam^t)\notag\\
&\hspace{2.6cm}-\frac{\tau}{2}\|\bW-\bW^t\|_F^2,\label{subW}\\
\bB^{t+1}&\in\arg\max_{\bB=\bB^T, \bB\in\mathcal{B}}\mathcal{L}_{\rho}(\y^{t+1},\boldsymbol{\gamma}^{t+1},\bW^{t+1},\bB,\bU^{t},\blam^t)\notag\\
&\hspace{2.8cm}-\frac{\xi}{2}\|\bB-\bB^t\|^2_{\text{sym}},\label{subB}\\
\bU^{t+1}&\in\arg\max~\mathcal{L}_{\rho}(\y^{t+1},\boldsymbol{\gamma}^{t+1},\bW^{t+1},\bB^{t+1},\bU,\blam^t),\label{subU}\\
\boldsymbol{\lambda}^{t+1}&=\boldsymbol{\lambda}^t+\rho((\mathbf{I}-\mathrm{i}Z_0\bB^{t+1})\bU^{t+1}-(\mathbf{I}+\mathrm{i}Z_0\bB^{t+1})\bH).\label{sublambda}
\end{align}
\end{subequations}
In \eqref{subB}, $\|\cdot\|_{\text{sym}}$ denotes the symmetric norm defined for a symmetric matrix. Specifically, for a symmetric matrix $\mathbf{X}\in\C^{M\times M},$ $\|\mathbf{X}\|_{\text{sym}}^2:=\sum_{i=1}^M|X_{i,i}|^2+\sum_{i<j}|X_{i,j}|^2,$which assigns equal weight to all distinct entries of $\mathbf{X}$. Different from classical ADMM, we incorporate proximal regularization terms in the updates of $\bW$ and $\bB$, controlled by the parameters $\tau$ and $\xi$, respectively. These proximal terms are essential for guaranteeing numerical stability and for enabling a rigorous theoretical convergence analysis of the proposed algorithm. Further details are provided  in Section \ref{sec:W} and Section \ref{sec:B}.  We next detail the update of each variable. 
\subsubsection{Update of $\y$} The $\y$-subproblem \eqref{suby} is unconstrained and quadratic, admitting a closed-form solution as  \begin{equation}\label{updatey} y_k^{t+1}=\frac{{(\u_k^t)}^\dagger\bG\bw_k^t}{\sum_{j=1}^K|{(\u_k^t)}^\dagger\bG\bw_j^t|^2+\sigma^2},~~~k=1,2,\dots, K.\end{equation}
\subsubsection{Update of $\boldsymbol{\gamma}$} The $\boldsymbol{\gamma}$-subproblem \eqref{subgamma} is strictly concave. By setting the derivative of $\tilde{R}(\y^{t+1},\boldsymbol{\gamma},\bW^t,\bU^t)$ with respect to $\boldsymbol{\gamma}$ to zero and solving the corresponding equation, it is easy to obtain the following closed-form solution:
\begin{equation}\label{updategamma}
\gamma_k^{t+1}=\frac{|(\u_k^t)^\dagger\bG\bw_k^t|^2}{\sum_{j\neq k}|(\u_k^t)^\dagger\bG\bw_j^t|^2+\sigma^2},~~k=1,2,\dots, K.
\end{equation}
\subsubsection{Update of $\bW$} \label{sec:W}
The $\bW$-subproblem can be written  as
 \begin{equation}\label{updatewk}
\min_{\|\bW\|_F^2\leq P_T}\,\sum_{k=1}^K\left(\bw_k^\dagger\bQ\bw_k-2\left<\bc_k,\bw_k\right>\right)+\frac{\tau}{2}\|\bW-\bW^t\|_F^2,
\end{equation}
where $\bQ=\sum_{j=1}^K(1+\gamma_j^{t+1})|y_j^{t+1}|^2\bG^\dagger\u_j^t({\u_j^t})^\dagger\bG$ and $\bc_k=y_k^{t+1}(1+\gamma_k^{t+1})\bG^\dagger\u_k^t$ (in the notations $\bQ$ and $\bc_k$, we omit their dependance on $t$ for simplicity).  Note that $\bQ\in \C^{N\times N}$ is a summation of $K$ rank-one matrices, and thus has a rank of at most $K$. This implies that when $N>K$, i.e., the number of transmit antennas is larger than the number of users in the system, $\bQ$ is rank deficient. The proximal term is introduced to ensure that the objective function of \eqref{updatewk} is  strongly convex, which is important {\color{black} for establishing the theoretical} convergence of the algorithm. 

Problem \eqref{updatewk} is a quadratically constrained quadratic programming (QCQP), which can be efficiently solved via a one-dimensional bisection search. Based  on the first-order optimality condition of \eqref{updatewk}, its solution  can be expressed as  
\begin{equation*}\label{wsolution}
\bw_k^{t+1}=\bigg(\bQ+\left(\frac{\tau}{2}+\eta^*\right)\mathbf{I}_N\bigg)^{-1}\left(\bc_k+\frac{\tau}{2}\bw_k^t\right),~k=1,\dots, K,\end{equation*}
where $\eta^*\geq 0$ is the Lagrange multiplier associated with the total transmit power constraint and is determined according to the complementary slackness condition. Specifically,  let 
$$\bw_k(\eta)=\left(\bQ+\left(\frac{\tau}{2}+\eta\right)\mathbf{I}_N\right)^{-1}\left(\bc_k+\frac{\tau}{2}\bw_k^t\right).$$ Then 
$\eta^*=0$ if $\sum_{k=1}^K\|\bw_k(0)\|^2\leq P_T$ and $\eta^*$ is the solution to 
\begin{equation}\label{powerP}
\sum_{k=1}^K\|\bw_k(\eta)\|^2=P_T
\end{equation} otherwise.
Equation \eqref{powerP} can be efficiently solved via a one-dimensional bisection search by rewriting it as  $$\sum_{n=1}^N\frac{\|\boldsymbol{\phi}_n\|^2}{(d_n+\tau/2+\eta)^2}=P_T,$$
where $\boldsymbol{\phi}_n^T$ is the $n$-th row of the matrix 
$\bU_Q^\dagger[\bc_1+\tau\bw_1^t/2,~\bc_2+\tau\bw_2^t/2,\dots,\bc_K+\tau\bw_K^t/2]$, and $\bQ=\bU_Q\bD_Q\bU_Q^\dagger$ is the eigenvalue decomposition of $\bQ$ with $\bD_Q=\text{diag}(d_1,d_2,\dots,d_N)$; see more details in \cite{WMMSE}.

\subsubsection{Update of $\bB$}\label{sec:B}
The $\bB$-subproblem \eqref{subB} has the following form:
\begin{equation}\label{updateB1}
\begin{aligned}
\hspace{-0.1cm}\min_{\bB=\bB^T\atop \bB\in\mathcal{B}}\hspace{-0.05cm}{\color{black}\frac{\rho}{2}}&\left\|\bB(\mathrm{i}Z_0\bU^t\hspace{-0.1cm}+\hspace{-0.05cm}\mathrm{i}Z_0\H)\hspace{-0.05cm}-\hspace{-0.1cm}\left(\bU^t\hspace{-0.1cm}-\hspace{-0.05cm}\H\hspace{-0.05cm}+\hspace{-0.05cm}\frac{\blam^{t}}{\rho}\hspace{-0.05cm}\right)\hspace{-0.05cm}\right\|^2\hspace{-0.15cm}+\hspace{-0.05cm}\frac{\xi}{2}\|\bB\hspace{-0.05cm}-\hspace{-0.05cm}\bB^t\|^2_{\text{sym}}.
\end{aligned}
\end{equation}
Since variable $\bB$ is real-valued, it is convenient to transform \eqref{updateB1} into the real space as
\begin{equation}\label{updateB2}
\min_{\bB=\bB^T, \bB\in\mathcal{B}}~~{\color{black}\frac{\rho}{2}}\|\bB\mathbf{M}-\mathbf{\Gamma}\|_F^2+\frac{\xi}{2}\|\bB-\bB^t\|_\text{sym}^2,
\end{equation}
where 
\begin{equation*}\label{def:MGamma}
\begin{aligned}
\bM&=[\RR(\mathrm{i}Z_0\bU^t+\mathrm{i}Z_0\H),~\I(\mathrm{i}Z_0\bU^t+\mathrm{i}Z_0\H)]\in\R^{M\times 2K}\\
\end{aligned}
\end{equation*}
and $$
\boldsymbol{\Gamma}=\left[\RR\left(\bU^t-\H+{\blam^{t}}/{\rho}\right),\I\left(\bU^t-\H+{\blam^{t}}/{\rho}\right)\right]\in\R^{M\times 2K}.
$$
Due to the symmetry of $\bB$ and the constraint $\bB\in\mathcal{B}$, the unknowns in $\bB$ are actually all its non-zero  upper {triangular}  elements. Let 
\begin{equation}\label{relation:xB}
\x=[\bB_{1,\mathcal{S}_1},\bB_{2,\mathcal{S}_2},\dots,\bB_{M,\mathcal{S}_M}]^T
\end{equation}\text{with }
$
\mathcal{S}_i=\{j\mid j\geq i,B_{i,j}\neq 0\}, 
$ i.e., 
  $\x$ collects all the non-zero elements in the upper triangular parts of $\bB$. Then problem \eqref{updateB2} can be equivalently expressed as an unconstrained quadratic programming over $\x$. Specifically, let  $\{\ba_{m}^T\}_{m=1}^M$ denote the rows of $\bM$, i.e., $\bM^T=[\ba_1,\ba_2,\dots,\ba_M]$, and let $\bb=\text{vec}(\boldsymbol{\Gamma}^T)$ and  $\mathcal{S}_i=\{i_1,i_2,\dots, i_{|\mathcal{S}_i|}\}.$
Then$$\|\bB\bM-\boldsymbol{\Gamma}\|_F^2=\|\bA\x-\bb\|_2^2,$$
where   \begin{equation*}\label{splitA}\bA=\left[\begin{matrix}\bA_{11}&\bA_{12}&\cdots&\bA_{1M}\\\bA_{21}&\bA_{22}&\cdots&\bA_{2M}\\\vdots&\vdots&\ddots&\vdots\\\bA_{M1}&\bA_{M2}&\cdots&\bA_{MM}
 \end{matrix}\right]\in\R^{2MK\times\sum_{i=1}^M|\mathcal{S}_i|}
 \end{equation*} 
 with $\bA_{ij}\in\R^{2K\times|\mathcal{S}_{\color{black}j}|}$ given by
\begin{equation}\label{def:Ai}
\bA_{ij}(:,q)=\left\{
\begin{aligned} 
\ba_{j_q},~~\,&\text{if }i=j;\\
\ba_j,~~~ &\text{if }i=j_q;\\
\mathbf{0},~~~~&\text{otherwise},~~~
\end{aligned}\right.	q=1,2,\dots, |\mathcal{S}_j|;
\end{equation}
here $\bA_{ij}(:,q)$ denotes the $q$-th column of $\bA_{i j}$. With the above notations, problem \eqref{updateB2} can be written as 

\begin{equation*}\label{updatex}\min_\x~~{\color{black}\frac{\rho}{2}}\|\bA\x-\bb\|^2_2+\frac{\xi}{2}\|\x-\x^t\|^2_2,
\end{equation*}
which admits a closed form solution as 
\begin{equation}\label{x:closedform}
\x^{t+1}=\left({\color{black}\rho}\bA^T\bA+\xi\mathbf{I}\right)^{-1}({\color{black}\rho}\bA^T\bb+\xi\x^t).
\end{equation}
Note that calculating \eqref{x:closedform} becomes computationally expensive when the number of columns in $\bA$ is large. This happens, e.g., for fully-connected RIS, in which case  the number of columns in $\bA$ is $M(M+1)/2$, resulting in a complexity of $\mathcal{O}(M^6)$ for calculating the matrix inversion in \eqref{x:closedform}. To address this issue, we adopt an  equivalent expression of \eqref{x:closedform} when the number of columns in $\bA$ exceeds the number of its rows, i.e., when $\sum_{i=1}^M|\mathcal{S}_i|\geq 2MK$, as follows:
\begin{equation}\label{x:closedform2}
\x^{t+1}=\frac{1}{\xi}\left(\mathbf{I}-{\color{black}\rho}\bA^T\left({\color{black}\rho}\bA\bA^T+\xi\mathbf{I}\right)^{-1}\bA\right)(\rho\bA^T\bb+\xi\x^t).
\end{equation}
With \eqref{x:closedform2}, we only need to compute a matrix inversion of dimension $2MK$ rather than $\sum_{i=1}^M|\mathcal{S}_i|$. 

After obtaining $\x^{t+1}$, we update $\bB^{t+1}$ according to the relationship between $\x$ and $\bB$ defined in \eqref{relation:xB}. Specifically, $\bB^{t+1}$ is  given by
 \begin{equation*}\label{Bsolution}
 B^{t+1}_{i,j}=\left\{
 \begin{aligned}
\x^{t+1}(\textstyle\sum_{q=1}^{i-1}|\mathcal{S}_q|+l),~~&\text{if }j=i_l;\\
\x^{t+1}(\textstyle\sum_{q=1}^{j-1}|\mathcal{S}_q|+l),~~&\text{if }i=j_l;\\
0,~~~~~~~~~~~&\text{otherwise},
 \end{aligned}\right.
 \end{equation*}
where we use the notation {$\x(n)$ to denote the $n$-th entry of $\x$.} 

We remark that the proximal term in \eqref{subB} is introduced to ensure that the problem is strongly convex (and thus admits a unique solution), which is essential for establishing the convergence of the proposed algorithm. In addition, it enforces the variable $\x$ (equivalently, $\bB$) to update gradually. Without this proximal term, the elements of $\x$ may undergo abrupt update and take on excessively large values, especially when $\bA$ has very small singular value, leading to numerical instability and degraded performance.  

\subsubsection{Update of $\bU$} The $\bU$-subproblem \eqref{subU} is separable in $\{\u_k\}_{1\leq k\leq K}$. For each $\u_k$, the problem is an unconstrained  convex quadratic programming, with the closed-form solution given by 
\begin{equation*}\label{updateu}
\begin{aligned}
\u_k^{t+1}=&\bigg((1+\gamma_k^{t+1})|y_k^{t+1}|^2\bG\bW^{t+1}(\bW^{t+1})^\dagger\bG^\dagger\\
&+\frac{\rho}{2}(\mathbf{I}+Z_0^2(\mathbf{B}^{t+1})^2)\bigg)^{-1}\bigg((1+\gamma_k^{t+1})(y_k^{t+1})^\dagger\bG\bw_k^{t+1}\\
&-\frac{1}{2}(\mathbf{I}+\mathrm{i}Z_0\bB^{t+1})\blam_k^t+\frac{\rho}{2}(\mathbf{I}+\mathrm{i}Z_0\mathbf{B}^{t+1})^{2}\h_k\bigg),\\
\end{aligned}
\end{equation*}
where we have used the fact that $(\mathbf{I}-\mathrm{i}Z_0\bB)^{\dagger}=\mathbf{I}+\mathrm{i}Z_0\bB$, and  $\blam_k$ denotes the $k$-th column of $\blam$.

\subsection{Complexity and Convergence Analysis}\label{sec:algorithm3b}
In this subsection, we give complexity and convergence analysis of the proposed pp-ADMM algorithm. 
\subsubsection{Complexity Analysis} The complexity of updating variables $\y$ and $\boldsymbol{\gamma}$ is $\mathcal{O}(MNK+K^2)$.  Updating variable $\bW$ requires a complexity of $\mathcal{O}(N^2K+N^3)$, where the first term comes from constructing $\bQ$ (note that $\bG^{\dagger}\bU^t$ has already been calculated via the update of $\y$ and $\boldsymbol{\gamma}$), and the second term comes from calculating the eigenvalue decomposition of $\bQ$.   The computational cost of updating variable $\bB$ is dominated by calculating the matrix inversion in  \eqref{x:closedform} or \eqref{x:closedform2}, whose complexity is $\mathcal{O}((\min\{\sum_{i=1}^M|\mathcal{S}_i|, 2MK\})^3)$. The complexity of updating variable $\bU$ is $\mathcal{O}(M^3K+MNK)$. Finally, updating $\blam$ requires a complexity of  $\mathcal{O}(M^2K)$. Since in practice the number of RIS elements $M$ is much larger than the number of transmit antennas $N$ and the number of users $K$, the total per-iteration complexity of the proposed pp-ADMM algorithm is $\mathcal{O}\left(M^3K+(\min\{\sum_{i=1}^M|\mathcal{S}_i|, 2MK\})^3\right)$.

\subsubsection{Convergence Analysis} Classical convergence results for ADMM-type algorithms are typically established for optimization problems with linear constraints \cite{boyd2011distributed}. In our work, however, ADMM is applied to tackle the bilinear constraint in \eqref{con:bilinear}, {\color{black}which falls outside the standard framework}.   Although a few works have also studied the convergence of ADMM on bilinear constrained problems \cite{bilinear1,bilinear2}, the  problems considered therein are much  simpler than  \eqref{eqn:reformulate2}.  {\color{black}Fortunately, by incorporating the additional proximal terms in \eqref{ppADMM} and leveraging the special structure of the bilinear constraint \eqref{con:bilinear}, we can establish the convergence of the proposed pp-ADMM algorithm. }
\begin{theorem}\label{converge}
Assume that the sequences $\{\bU^t\}_{t\geq 0}$ and $\{\blam^t\}_{t\geq 0}$ generated by the pp-ADMM algorithm in \eqref{ppADMM} are bounded. 
Then there exists $\rho_0>0$ such that when $\rho>\rho_0$, any limit point $\left(\y^*,\boldsymbol{\gamma}^*,\bW^*,\bB^{*},\bU^*\right)$ of $(\y^{t},\boldsymbol{\gamma}^{t},\bW^{t},$ $\bB^{t},\bU^{t})$   is a stationary point of problem \eqref{eqn:reformulate2}.
\end{theorem}
\begin{proof}
See Appendix B in the supplemental material.
\end{proof}
{\color{black}The boundedness assumptions on the sequences $\{\bU^t\}_{t\geq 0}$ and  $\{\blam^t\}_{t\geq 0}$ are imposed for technical reasons. Although a rigorous proof of these boundedness properties is currently unavailable, such boundedness is observed in all of our numerical experiments.} 


\section{Transmit Power minimization problem}\label{sec:4}
In this section, we investigate the transmit power minimization problem in \eqref{min_power}.\vspace{-0.2cm} 
\subsection{An Equivalent Formulation}\label{sec2:algorithm1}
Similar to Section \ref{sec:algorithm1}, we introduce an auxiliary variable $\bU=\bthe^\dagger\bH$ to deal with the matrix inversion in constraint \eqref{con:theta2}, which yields 
\begin{subequations}\label{minpower2}
\begin{align}
\min_{\bW,\bB,\bU}~&\|\bW\|_F^2\\
\text{s.t. }~~&\text{SINR}_k(\u_k,\bW)\geq\Gamma_k,~k=1,2,\dots, K,\label{con:QoS2}\\
&(\mathbf{I}-\mathrm{i}Z_0\bB)\bU=(\mathbf{I}+\mathrm{i}Z_0\bB)\bH,\label{23c}\\
&\bB=\bB^T,~\bB\in\mathcal{B}.\label{23d}
\end{align}
\end{subequations}
To tackle the QoS constraint in \eqref{con:QoS2}, we  introduce another auxiliary variable  ${\color{black}\bV}=\bU^\dagger\bG\bW\in\C^{K\times K}$. Then  problem \eqref{minpower2} can be transformed into the following equivalent form: 
\begin{subequations}\label{minpower3}
\begin{align}
\min_{\bW,\bB,\bU,\bV}\,&\|\bW\|_F^2\\
\text{s.t. }~~~\,&V_{k,k}\geq\sqrt{\Gamma_k({\sum_{j\neq k}|{\color{black}V_{k,j}}|^2+\sigma^2})},~k=1,2,\dots, K,\label{con:setY}\\
&\I({\color{black}V_{k,k}})=0,~k=1,2,\dots, K, \label{consetY2}\\
&{\color{black}\bV}=\bU^\dagger\bG\bW,\label{con:bilinearY}\\
&\eqref{23c} \text{ and }\eqref{23d}.
\end{align}
\end{subequations}
Note that in \eqref{minpower3}, the non-convex QoS constraint 
\begin{equation}\label{QoS}
|{\color{black}V_{k,k}}|^2\geq\Gamma_k({\sum_{j\neq k}|{\color{black}V_{k,j}}|^2+\sigma^2}),~k=1,2,\dots, K
\end{equation}
has been replaced by the convex constraints in \eqref{con:setY} and \eqref{consetY2}, where   \eqref{con:setY} can be formulated as a second-order cone (SOC) constraint. The rationale behind this is that,  given any feasible pair {$(\bW,{\color{black}\bV})$} satisfying \eqref{con:bilinearY} and \eqref{QoS}, we can rotate $\{\bw_k\}_{1\leq k\leq K}$ to obtain a new feasible pair {$(\bar{\bW},{\color{black}\bar{\bV}})$} such that $\{{\color{black}\bar{V}_{k,k}}\}_{1\leq k\leq K}$ is nonnegative (i.e., {$(\bar{\bW},{\color{black}\bar{\bV}})$} satisfies \eqref{con:setY}\,--\,\eqref{con:bilinearY}) and the  objective value remains the same. This idea is inspired by the classical work \cite{SOCP}, in which a similar transformation was applied to reformulate the traditional transmit power minimization problem as a second-order cone programming (SOCP). 

Now, the complicated QoS constraint in \eqref{con:QoS2} has been reformulated as a bilinear constraint \eqref{con:bilinearY} and two simple constraints, \eqref{con:setY} and \eqref{consetY2}, on ${\color{black}\bV}$ that allow for efficient projection. 
To be more specific, the projection onto set
$${\color{black}\mathcal{V}}:=\left\{{\color{black}\bV}\in\C^{K\times K}\mid{\color{black}\bV} \text{ satisfies \eqref{con:setY}} \text{ and } \eqref{consetY2}\right\}$$
 can be efficiently computed by solving $K$ one-dimensional quartic equations; see further details below.
 
\subsection{pp-ADMM Algorithm for Solving \eqref{minpower3}}
Problem \eqref{minpower3} has a simple objective function with separable constraints on ${\color{black}\bV}$ and $\bB$, as well as  two bilinear constraints that couple variables $(\bB,\bU,\bW,\bV)$, which is suitable to be solved via an ADMM framework.  Next, we custom-build  an efficient pp-ADMM algorithm for solving \eqref{minpower3}. 

First,  the augmented Lagrangian function of \eqref{minpower3} is given by
\begin{equation}\label{Lrho}\begin{aligned}
&\mathcal{L}_{\boldsymbol{\rho}}({\color{black}\bV},\bW,\bB,\bU,\blam,\boldsymbol{\mu})\\
=\,&\|\bW\|_F^2+\left<\blam,(\mathbf{I}-\mathrm{i}Z_0\mathbf{B})\bU-(\mathbf{I}+\mathrm{i}Z_0\bB)\bH\right>\\
&+\frac{\rho_{\lambda}}{2}\|(\mathbf{I}-\mathrm{i}Z_0\mathbf{B})\bU-(\mathbf{I}+\mathrm{i}Z_0\bB)\bH\|_F^2\\
&+\left<\boldsymbol{\mu},{\color{black}\bV}-\bU^\dagger\bG\bW\right>+\frac{\rho_\mu}{2}\left\|{\color{black}\bV}-\bU^\dagger\bG\bW\right\|_F^2,
\end{aligned}
\end{equation}
where $(\blam,\bmu)$ are the Lagrange multipliers associated with the bilinear constraints \eqref{23c} and \eqref{con:bilinearY}, respectively, and $\boldsymbol{\rho}=(\rho_\lambda,\rho_\mu)$ are the corresponding penalty parameters. 

The proposed pp-ADMM algorithm is given below:
\begin{subequations}\label{ppADMM2}
\begin{align}
{\bV}^{t+1}&\in\arg\min_{\bV\in\mathcal{V}}~\mathcal{L}_{\boldsymbol{\rho}}(\bV,\bW ^t,\bB^t,\bU^t,\blam^t,\boldsymbol{\mu}^t),\label{subY2}\\
\bW^{t+1}&\in\arg\min~\mathcal{L}_{\boldsymbol{\rho}}({\color{black}\bV}^{t+1},\bW,\bB^t,\bU^t,\blam^t,\boldsymbol{\mu}^t),\label{subW2}\\
\bB^{t+1}&\in\arg\min_{\bB=\bB^T,\bB\in\mathcal{B}}\mathcal{L}_{\boldsymbol{\rho}}({\color{black}\bV}^{t+1},\bW^{t+1},\bB,\bU^t,\blam^t,\boldsymbol{\mu}^t)\notag\\
&\hspace{2.8cm}+\frac{\xi}{2}\|\bB-\bB^t\|^2_{\text{sym}},\label{subB2}\\
\bU^{t+1}&\in\arg\min~\mathcal{L}_{\boldsymbol{\rho}}({\color{black}\bV}^{t+1},\bW^{t+1},\bB^{t+1},\bU,\blam^t,\boldsymbol{\mu}^t),\label{subU2}\\
\blam^{t+1}&=\blam^t\hspace{-0.05cm}+\hspace{-0.05cm}\rho_\lambda\left((\mathbf{I}-\mathrm{i}Z_0\bB^{t+1})\bU^{t+1}-(\mathbf{I}+\mathrm{i}Z_0\bB^{t+1})\bH\right),\label{sublambda2}\\
\boldsymbol{\mu}^{t+1}&=\boldsymbol{\mu}^{t}+\rho_\mu(\mathbf{V}^{t+1}-(\bU^{t+1})^\dagger\bG\bW^{t+1}).\label{submu2}
\end{align}
\end{subequations}

 The only difference between the above algorithm and  classical ADMM is  an extra proximal term for updating $\bB$, which is introduced to ensure the stability and convergence of the algorithm, as discussed in Section \ref{sec:algorithm2}. The update of $\bB$ follows exactly the same procedure as in \eqref{updateB1}. 
In addition, the $\bW$- and $\bU$-subproblems in \eqref{subW2} and \eqref{subU2} are both unconstrained convex quadratic programming, which admit closed-form solutions.   We omit the details for brevity. Next, we discuss the solution of the {\color{black}$\bV$}-subproblem \eqref{subY2}.

According to \eqref{Lrho} and \eqref{subY2},  $\bV^{t+1}$ is the projection of $(\bU^t)^\dagger\bG\bW^t-\boldsymbol{\mu}^t/\rho_\mu$ onto set $\mathcal{V}$. 
 Since constraint \eqref{con:setY} is separable across the rows of $\bV$, the $\bV$-subproblem can be divided into  $K$  independent problems as follows (where $k=1,2,\dots, K$):
\begin{equation}\label{y-problem}
\begin{aligned}
\min_{\bV_{k,:}}~&\sum_{j=1}^K\left|V_{k,j}-(\u_k^t)^{\dagger}\bG\bw_j^{t}+\frac{\mu_{k,j}^t}{\rho_\mu}\right|^2\\
\text{s.t. }~&V_{k,k}\geq\sqrt{\Gamma_k(\sum_{j\neq k}|V_{k,j}|^2+\sigma^2)},~~\I(V_{k,k})=0.
\end{aligned}
\end{equation}
Clearly, the constraint in \eqref{y-problem} depends only on the magnitude of  $\{V_{k,j}\}_{ j\neq k}$ and is irrelevant to their phases.  To minimize the objective function, the phase of $V_{k,j}$ should align with that of  {\small$(\u_k^t)^\dagger\bG\bw_j^t-{\mu}_{k,j}^t/\rho_\mu$}, where $j\neq k$. Due to this, and by further squaring the constraint and noting that $V_{k,k}$ is real-valued and positive, 
we can rewrite problem \eqref{y-problem} into the following equivalent  form:
\begin{equation}\label{y-problem2}
\begin{aligned}
\min_{\{r_{k,j}\}_{j=1}^K}~&\sum_{j=1}^K(r_{k,j}-a_{k,j})^2\\
\text{s.t. }~\,~~&r_{k,k}^2\geq \Gamma_k(\sum_{j\neq k}r_{k,j}^2+\sigma^2),~~r_{k,k}\geq 0.
\end{aligned}
\end{equation}
where
$$r_{k,j}=
\left\{
\begin{aligned}
|{\color{black}V_{k,j}}|,~~&\text{if } j\neq k;\\
{\color{black}V_{k,k}},\,~~&\text{if }j=k,
\end{aligned}\right.$$ and 
$$
a_{k,j}=\left\{
\begin{aligned}
|(\u_k^t)^{\dagger}\bG\bw_j^{t}-{\mu_{k,j}^t}/{\rho_\mu}|,~~~\text{if }j\neq k;\\
\RR((\u_k^t)^{\dagger}\bG\bw_k^{t}-{\mu_{k,k}^t}/{\rho_\mu}),~~\text{if }j=k.
\end{aligned}\right.$$
 Strictly speaking, there should be nonnegative constraints on  $\{r_{k,j}\}_{j\neq k}$ in \eqref{y-problem2}, as they represent the magnitudes of $\{{V}_{k,j}\}_{j\neq  k}$. However, we  can  ignore these constraints since the optimal solution to \eqref{y-problem2} will naturally be nonnegative due to the nonnegativity of $\{a_{k,j}\}_{j\neq  k}$.
By analyzing the KKT condition of \eqref{y-problem2}, we claim that the optimal solution to \eqref{y-problem2} is given by 
\begin{equation}\label{rjstar}
r_{k,j}^*=\left\{
\begin{aligned}
\left({\Gamma_k}\sum_{j\neq k} \frac{a_{k,j}^2}{(1+\Gamma_k)^2}+\Gamma_k\sigma^2\right)^{\frac{1}{2}}\hspace{-0.2cm},~&\text{if }j=k,~ a_{k,k}=0;\\
\frac{a_{k,j}}{1-\eta^*},\hspace{2.2cm}&\text{if }j=k,~a_{k,k}\neq 0;\\
\frac{a_{k,j}}{1+\Gamma_k\eta^*},\hspace{2cm}&\text{if }j\neq k,
\end{aligned}\right.
\end{equation}
where
\begin{equation}\label{eta2}
\eta^*=\left\{
\begin{aligned}
0,~~~~~&\text{if }a_{k,k}\geq({\Gamma_k\sum_{j\neq k} a_{k,j}^2+\Gamma_k\sigma^2})^{{1}/{2}};\\
\eta_1^*,~~~~&\text{if }0<a_{k,k}<({\Gamma_k\sum_{j\neq k} a_{k,j}^2+\Gamma_k\sigma^2})^{{1}/{2}};\\
1,~~~~~&\text{if }a_{k,k}=0;\\
\eta_2^*,~~~~&\text{if }a_{k,k}<0,
\end{aligned}\right.
\end{equation}
with $\eta_1^*\in(0,1)$ and $\eta_2^*\in(1,\infty)$ being the unique roots of the following equation in the intervals  $(0,1)$ and $(1,\infty)$, respectively:
\begin{equation}\label{eta}
\frac{\Gamma_k\sum_{j\neq k}a_{k,j}^2}{(1+\Gamma_k\eta)^2}-\frac{a_{k,k}^2}{(1-\eta)^2}+\Gamma_k\sigma^2=0.
\end{equation}
 A rigorous proof of \eqref{rjstar} is given in {\color{black} Appendix A in the supplemental material}. Eq. \eqref{eta} can  be further transformed into a quartic equation and solved in closed-form.  Combining the above, we get
\begin{equation*}\label{Ysolution_2}
V_{k,j}^{t+1}=\left\{
\begin{aligned}
&r_{k,j}^*e^{\mathrm{i}\arg\left((\u_k^t)^\dagger\bG\bw_j^t-{\mu}_{k,j}^t/\rho_\mu\right)},~\text{if }j\neq k;\\
&r_{k,k}^*,~\hspace{3.45cm}\text{if }j=k.
\end{aligned}\right.
\end{equation*}
We remark that the technique of introducing an auxiliary variable \( \mathbf{V} \) to deal with the SINR constraint can be applied more broadly to optimization problems involving SINR constraints beyond the BD-RIS scenario. Furthermore, the resulting closed-form projection onto $\mathcal{V}$, derived via a careful KKT-based analysis, is scenario-independent and can be directly extended to other related problems. 
\subsection{Complexity and Convergence Analysis}\label{sec:algorithm3}
We next give complexity and convergence analysis of the above algorithm. 
\subsubsection{Complexity Analysis} The complexity for updating variables $\bV$, $\bW$, $\bB$, $\bU$, $\blam$, and $\boldsymbol{\mu}$ are $\mathcal{O}(K^2), \mathcal{O}(N^2K+NK^2+N^3),$ $\mathcal{O}((\min\{\sum_{i=1}^M|\mathcal{S}_i|, 2MK\})^3)$,  $\mathcal{O}(M^3+M^2K+MNK+MK^2),$ $\mathcal{O}(M^2K),$ and $\mathcal{O}(MNK)$, respectively\footnote{We remark that the complexity of calculating $(\bU^t)^{\dagger}\bG\bW^t$ is not included in the complexity of updating $\bV^{t+1}$, as this computation has already been performed in updating $\boldsymbol{\mu}^{t}$.  Similarly, the complexity of calculating $\bG^\dagger\bU^t$ is  excluded from the complexity of updating $\bW^{t+1}$.}.  To conclude, the total per-iteration complexity of the pp-ADMM algorithm in \eqref{ppADMM2} is $\mathcal{O}\left(M^3+(\min\{\sum_{i=1}^M|\mathcal{S}_i|, 2MK\})^3\right)$.
\subsubsection{Convergence Analysis}Similar to Theorem \ref{converge}, we can establish the convergence of the pp-ADMM algorithm in \eqref{ppADMM2}. 
\begin{theorem}\label{converge2}
Assume that 
\begin{itemize}
\item[(a)] The sequences $\{\bW^t\}_{t\geq 0}$, $\{\bU^t\}_{t\geq 0}$,  $\{\blam^t\}_{t\geq 0}$, and $\{\boldsymbol{\mu}^t\}_{t\geq 0}$ generated by the pp-ADMM algorithm in \eqref{ppADMM2} are bounded.
\item[(b)] The minimum singular value of $\{\bG^{\dagger}\bU^t\}_{t\geq 0}$ has a uniformly positive lower bound.  
\end{itemize}
 Then there exists $c_0>0$ and $\rho_0>0$ such that when $\rho_\lambda>c_0\rho_{\mu}>\rho_0$, any limit point $\left(\bV^*,\bW^*,\bB^{*},\bU^*\right)$ of $(\bV^{t},\bW^{t},$ $\bB^{t},\bU^{t})$   is a stationary point of problem \eqref{minpower3}.
\end{theorem}
\hspace{-0.35cm}
 {\color{black}Note that in problem \eqref{minpower3}, there are two bilinear constraints \eqref{23c} and \eqref{con:bilinearY} that couple variables $(\bU,\bB,\bW,\bV)$. Due to this, the assumptions required in Theorem \ref{converge2} are more restrictive than those in Theorem \ref{converge}.} The proof of Theorem \ref{converge2} is given in Appendix C in the supplemental material. 

\begin{remark}\label{remark}
It is worth mentioning a related work \cite{PDD}, in which  a penalty dual decomposition (PDD)-based optimization framework was developed to solve both sum-rate maximization and transmit power minimization problems for BD-RIS-aided MU-MISO communication systems. Our proposed approach outperforms \cite{PDD} in that it is applicable to any BD-RIS architecture, while the framework in \cite{PDD} is limited to only fully- and group-connected architectures.
 
 
{Moreover, compared with the PDD algorithm in \cite{PDD}, the proposed  algorithms are computationally more efficient due to both the simpler algorithmic structure and the reduced per-iteration complexity. First, the PDD algorithm in \cite{PDD} is double loop, where an inner loop is required to tackle the symmetric and unitary constraints on the scattering matrix, whereas the proposed algorithms are single-loop. This avoids nested inner-loop updates and can reduce the practical computational burden. Second, for the sum-rate maximization problem, the per-iteration complexity of the proposed approach is significantly lower than that of PDD. 
   Taking the fully-connected RIS as an example. The scattering-matrix update in PDD requires solving a linear system whose dimension is $\mathcal{O}(M^2)$ \cite[Eq. (59)]{PDD}, leading to a per-iteration complexity of $\mathcal{O}(M^6)$. In contrast, the proposed approach reduces the per-iteration complexity to $\mathcal{O}(K^3M^3)$ by replacing $\bthe\in\C^{M\times M}$ with the auxiliary variable $\bU\in\C^{M\times K}$ and leveraging \eqref{x:closedform2}. Third, for the transmit power minimization problem, the PDD algorithm formulates an SOCP in each scattering matrix subproblem to handle the QoS constraints, where the variable dimension scales as $\mathcal{O}(M^2)$. The subproblem is solved by  calling CVX, which is computationally expensive. In contrast, the proposed approach introduces appropriate auxiliary variables to simplify the QoS constraints, so that the resulting subproblems admit closed-form updates. The practical computational advantage of the proposed approaches is further verified by the CPU-time results in Section \ref{sec:comparison_algorithm}.} 
 \end{remark}
 \section{Generalizations}\label{sec:generalization}
In previous sections, we focus on sum-rate maximization and transmit power minimization problems for MU-MISO systems, assuming  {perfect instantaneous  CSI and a BD-RIS operating in the reflective mode} with infinite-resolution admittances (i.e., the entries of $\bB$ are continuous). However, the technique for dealing with  BD-RIS-related variables (i.e., $\bthe$ and $\bB$) is general, and the proposed optimization framework can be extended {\color{black} in several aspects}. 
In this section, we explore the broader applicability of our proposed approach.
 \subsection{General Utility Functions}
Consider a utility optimization problem in the following  form:
\begin{equation}\label{general}
\begin{aligned}
\max_{\bW,\bthe,\bB}~&F(\{\text{SINR}_{k}(\bthe,\bW)\}_{k=1}^K)\\
\text{s.t. }~~ &\bthe=\left(\mathbf{I}+\mathrm{i}Z_0\mathbf{B}\right)^{-1}\left(\mathbf{I}-\mathrm{i}Z_0\mathbf{B}\right),
\\
~~&\bB=\bB^{T},~\bB\in\mathcal{B},\\
&\bW\in\mathcal{W},
\end{aligned}
\end{equation}
where $F(\{\text{SINR}_{k}(\bthe,\bW)\}_{k=1}^K)$ is a general utility function and $\bW\in\mathcal{W}$ represents the constraint imposed on the transmitter side (e.g., the transmit power constraint).  The above model encompasses the sum-rate maximization problem in \eqref{multiuser} with 
$F(\{\text{SINR}_{k}(\bthe,\bW)\}_{k=1}^K)=\sum_{k=1}^K\log\left(1+\text{SINR}_{k}(\bthe,\bW)\right)$ and $\mathcal{W}=\{\|\bW\|_F^2\leq P_T\}$, and the transmit power minimization problem in \eqref{min_power} with
 $$
 \begin{aligned}
& F(\{\text{SINR}_{k}(\bthe,\bW)\}_{k=1}^K)\\
&=-\|\bW\|_F^2-\sum_{k=1}^K\mathbb{I}_{\{\text{\normalfont{SINR}}_k(\bthe,\bW)\geq \gamma_k\}}(\bthe,\bW)
 \end{aligned}$$
 and $\mathcal{W}=\C^{N\times K},$ where $\mathbb{I}_{\mathcal{X}}(x)$ refers to the indicator function of set $\mathcal{X}$, defined as $\mathbb{I}_{\mathcal{X}}(x)=0$ if $x\in\mathcal{X}$, and  $\mathbb{I}_{\mathcal{X}}(x)=\infty$ otherwise.  

Inspired by the previous sections, we extract the following procedures for solving \eqref{general}.

\emph{\text{Step 1}}: 
 The first step deals with the difficulties introduced by BD-RIS. Specifically, we introduce auxiliary variables $\u_k=\bthe^\dagger \h_k,~k=1,2,\dots, K$, to transform \eqref{general} into the following equivalent form:
\begin{subequations}\label{general2}
\begin{align}
\max_{{\bW,\bB,\bU}}~&F(\{\text{SINR}_{k}(\u_k,\bW)\}_{k=1}^K)\label{general2:0}\\
\text{s.t. }~~ &\left(\mathbf{I}-\mathrm{i}Z_0\mathbf{B}\right)\bU=\left(\mathbf{I}+\mathrm{i}Z_0\mathbf{B}\right)\bH,\label{general2:1}\\
~~&\bB=\bB^{T},~\bB\in\mathcal{B},\label{general2:2}\\
&\bW\in\mathcal{W}\label{general2:3},
\end{align}
\end{subequations}
where $\text{SINR}_k(\u_k,\bW)$ is given in \eqref{SINR:ukW}. This transformation is the key of our proposed framework, which offers a simpler formulation with reduced dimensionality; see the discussions below \eqref{eqn:reformulate1}.

\emph{\text{Step 2}}: The second step is to further simplify the objective function in \eqref{general2:0},  typically by introducing appropriate auxiliary variables. 
The approach for this step relies on the specific structure of the utility function under consideration. 
For example,  FP is employed  to simplify the sum-rate objective function in \eqref{FP:sumrate}, and an  auxiliary variable ${\color{black}\bV}=\bU^\dagger \bG\bW$ is introduced in \eqref{minpower3} to simplify the QoS requirement. 
 
\emph{\text{Step 3}}: The final step is to employ an ADMM framework to deal with the bilinear constraint \eqref{general2:1}. 
In particular, appropriate proximal terms can be introduced for specific variables to ensure {\color{black}theoretical convergence and improve numerical  stability}. For efficient implementation, the objective function formulated in Step 2 should be simple with respect to each variable block.

In addition to the sum-rate maximization and transmit power minimization problems considered in the previous sections, we  next briefly discuss the application of the above framework  to two other widely used  utilities: max-min fairness and energy efficiency maximization.
\subsubsection{Max-min fairness} Max-min fairness aims to maximize the utility of the user with the worst performance. Mathematically, the utility function is  given by   
$$F(\{\text{SINR}_k(\u_k,\bW)\}_{k=1}^K)=\min_k~\text{SINR}_k(\u_k,\bW).$$
 Introducing  an auxiliary variable $\Gamma=\min_k\,\text{SINR}_k(\u_k,\bW)$ and employing the transformation in {Step 1}, the max-min fairness model can be written as 
\begin{equation}\label{maxmin}
\begin{aligned}
\max_{{\bW,\bB,\bU}, \Gamma}~&\Gamma\\
\text{s.t. }~~~ &\text{SINR}_k(\u_k,\bW)\geq \Gamma,~~k=1,2,\dots, K,\\
&\eqref{con:bilinear} - \eqref{con:power2},
\end{aligned}
\end{equation}
where we consider the total transmit power constraint at the BS, i.e., $\mathcal{W}=\{\bW\mid\|\bW\|_F^2\leq P_T\}$. The above model is similar to \eqref{minpower3},  but with the following two differences. First, the total transmit power is imposed as a constraint in \eqref{maxmin}, while  it appears as  the objective function in \eqref{minpower2}. Second, the SINR threshold $\Gamma$ in  \eqref{maxmin} is treated as a variable, while $\{\Gamma_k\}_{k=1}^K$ in \eqref{minpower2} are fixed  problem parameters. As in \eqref{minpower2}, we further simplify the QoS constraint with an  auxiliary variable ${\color{black}\bV}=\bU^\dagger\bG\bW$,  after which an ADMM framework similar to \eqref{ppADMM2} can be employed to solve the problem.  
{The only differences with \eqref{ppADMM2} lie in the $\bW$-subproblem and the $\Gamma$-subproblem. In particular, the $\bW$-subproblem is a QCQP and can be efficiently solved via a bisection search, following similar steps as in Section \ref{sec:W}. 
At the $(t+1)$-th iteration, the $\Gamma$-subproblem has the following form:
\begin{equation*}
\begin{aligned}
\max_{\Gamma}~&\Gamma\\
\text{s.t.}~~&V_{k,k}^t\geq\sqrt{\Gamma(\sum_{j\neq k}|V_{k,j}^t|^2+\sigma^2)},~k=1,2,\dots, K,
\end{aligned}
\end{equation*}
and  is updated as 
$$\Gamma^{t+1}=\min_k\left\{\frac{|V_{k,k}^t|^2}{\sum_{j\neq k}|V_{k,j}^t|^2+\sigma^2}\right\}.$$ }

\subsubsection{Energy efficiency maximization} The energy efficiency of the system is given as 
$$F(\{\text{SINR}_k(\u_k,\bW)\}_{k=1}^K)=\frac{\sum_{k=1}^K\log\left(1+\text{\normalfont{SINR}}_k(\u_k, \bW)\right)}{\mu_0\|\bW\|_F^2+P_c},$$
where $\mu_0$ is the power amplifier efficiency and  $P_c$ represents the total circuit power. Applying the FP technique, 
 the energy efficiency maximization problem can be equivalently formulated  as   
\begin{equation}\label{effiency}
\begin{aligned}
\max_{{\bW, \bB, \bU}, v}\,&v\sum_{k=1}^K\log(1+\text{\normalfont{SINR}}_k(\u_k, \bW))-\frac{v^2}{2}(\mu_0\|\bW\|_F^2+P_c)^2\\
\text{s.t. }~~ &\eqref{con:bilinear} - \eqref{con:power2}.
\end{aligned}
\end{equation}
{Specifically, given $(\bU, \bW)$, the optimal $v$ to problem \eqref{effiency} is 
\begin{equation}\label{update:v}
v=\frac{\sum_{k=1}^K\log(1+\text{\normalfont{SINR}}_k(\u_k, \bW))}{(\mu_0\|\bW\|_F^2+P_c)^2}.
\end{equation}
Substituting this optimal $v$ into the objective function of \eqref{effiency} yields half the square of the energy efficiency. Since the energy efficiency is nonnegative, maximizing it is equivalent to maximizing its square,  hence \eqref{effiency} is equivalent to the original energy efficiency maximization problem.   }

The sum-rate expression in the objective function can further be simplified as in \eqref{FP:sumrate},  enabling the application of an efficient ADMM framework {similar to \eqref{ppADMM}. The corresponding algorithm differs from \eqref{ppADMM} only in the formula of the  $\bW$-subproblem and the inclusion of an additional variable $v$. At each iteration, the variable $v$ is updated in closed-form as in \eqref{update:v}, which is obtained by maximizing the objective function with all other variables fixed. Due to the additional power-consumption term in the objective function, the $\mathbf{W}$-subproblem involves a quartic objective function under the same quadratic transmit power constraint, and is therefore different from the QCQP in \eqref{subW}. Specifically, the $\bW$-subproblem takes the form 
\begin{equation}
\begin{aligned}
\min_{\|\bW\|_F^2\leq P_T} &\sum_{k=1}^K\big(\bw_k^\dagger\bQ\bw_k-2\langle \mathbf{c}_k,\mathbf{w}_k\rangle\big)
\\
&+\frac{v^t}{2}(\mu_0\|\bW\|_F^2+P_c)^2+\frac{\tau}{2v^t}\|\bW-\bW^t\|_F^2,
\end{aligned}
\end{equation} 
where $\bQ$ and $\{\mathbf{c}_k\}$ are defined below \eqref{updatewk}. The above problem can still be efficiently solved using a one-dimensional bisection search. For brevity, we omit the detailed derivations here and refer interested readers to Appendix D of the supplemental material.}
  \subsection{MU-MIMO Systems}
  The proposed framework can also be extended to MU-MIMO systems. Denote $\H_k\in\C^{N_k\times M}$ as the channel from BD-RIS to the $k$-th user, where $N_k$ represents the number of antennas at the $k$-th user. 
  In addition, let   $\s_k\in\C^{d_k\times 1}$ be the data streams intended for the $k$-th user, and  let $\bW_k\in\C^{N\times d_k}$ denote the corresponding beamforming matrix.   The performance of the $k$-th user is typically characterized by its achievable rate, given by
 $$
  \begin{aligned}
R_k = \log\det \bigg( \mathbf{I} + 
& \, \mathbf{H}_{k,\text{eff}} \mathbf{W}_k \mathbf{W}_k^\dagger \mathbf{H}_{k,\text{eff}}^\dagger \nonumber \\
& \times \bigg( \sum_{j \neq k} \mathbf{H}_{k,\text{eff}} \mathbf{W}_j \mathbf{W}_j^\dagger \mathbf{H}_{k,\text{eff}}^\dagger + \sigma^2 \mathbf{I} \bigg)^{-1} \bigg),
\end{aligned}
$$
where $\H_{k,\text{eff}}:=\bH_k\bthe\bG\in\C^{N_k\times N}$ represents the effective channel from the BS to user $k$. Following a similar technique as in the MU-MISO case, we introduce auxiliary variables
\begin{equation}\label{Uk}
\bU_k=(\H_k\bthe)^\dagger\in\bC^{M\times N_k},~k=1,2,\dots, K,
\end{equation} with which the effective channel can be written as $\bH_{\text{eff,k}}=\bU_k^\dagger\bG,$ and the constraint $\bthe=\left(\mathbf{I}+\mathrm{i}Z_0\mathbf{B}\right)^{-1}\left(\mathbf{I}-\mathrm{i}Z_0\mathbf{B}\right)$ is transformed into 
\begin{equation}\label{bilinear:U}\left(\mathbf{I}-\mathrm{i}Z_0\mathbf{B}\right)\bU_k=\left(\mathbf{I}+\mathrm{i}Z_0\mathbf{B}\right)\bH_k^\dagger,~k=1,2,\dots, K.\end{equation}
The difference from the MU-MISO case is that the auxiliary variable introduced for each user, i.e., $\bU_k$, is now a matrix rather than a vector, due to 
 the presence of multiple antennas at the user side. Nevertheless, $\bU_k$ remains related to $\bB$ through a bilinear constraint as shown in \eqref{bilinear:U}, which can still be tackled using an ADMM framework. For example, the sum-rate maximization problem for MU-MIMO systems can be solved by first applying the FP technique (in matrix form) \cite{FP3} to simplify the objective function and then employing a similar pp-ADMM algorithm as in \eqref{ppADMM}.
\subsection{Statistical CSI}
In previous sections, we assume perfect instantaneous CSI to focus on the proposed architecture-independent BD-RIS optimization framework. The framework can nevertheless be extended to statistical CSI settings. Under statistical CSI, the instantaneous utility can be replaced by its expectation over channel realizations. For example, the ergodic sum-rate maximization problem reads
\[
\begin{aligned}
\max_{\bW,\boldsymbol{\Theta},\bB}\quad
&\mathbb{E}_{\mathbf H,\mathbf G}
\left[
\sum_{k=1}^{K}
\log\left(1+\text{SINR}_k(\boldsymbol{\Theta},\mathbf W;\mathbf H,\mathbf G)\right)
\right]\\
\text{s.t. }\quad~ &\eqref{con:theta} - \eqref{con:power},
\end{aligned}
\]
where $\text{SINR}_k(\boldsymbol{\Theta},\mathbf W;\mathbf H,\mathbf G)$ denotes the SINR of user $k$ under a given channel realization $(\mathbf H,\mathbf G)$; see \eqref{SINR}. A straightforward way to handle the expectation in the objective function is to approximate it using sample-average approximation, which gives the following approximate problem:
\[
\begin{aligned}
\max_{\bW,\boldsymbol{\Theta}, \bB}\quad
&\frac{1}{N_s}\sum_{s=1}^{N_s}\sum_{k=1}^{K}
\log\left(1+\text{SINR}_k(\boldsymbol{\Theta},\mathbf W;\mathbf H^{(s)},\mathbf G^{(s)})\right)\\
\text{s.t. }~\quad &\eqref{con:theta} - \eqref{con:power},
\end{aligned}
\]
where $N_s$ is the number of samples. 
For each channel sample $s$, one can introduce
\(
\mathbf U^{(s)}=\boldsymbol{\Theta}^{\dagger}\mathbf H^{(s)},
\)
leading to
\[
(\mathbf I-\mathrm{i}Z_0\mathbf B)\mathbf U^{(s)}
=
(\mathbf I+\mathrm{i}Z_0\mathbf B)\mathbf H^{(s)},
\qquad s=1,\ldots,N_s.
\]
 Hence, the resulting sample-average problem preserves the same bilinear structure as the instantaneous-CSI formulation and can be handled by a similar pp-ADMM framework.

The sample-average approximation provides a direct way to approximate the ergodic sum-rate maximization problem. Although this formulation preserves the structure of the proposed framework, it introduces multiple channel samples and sample-dependent auxiliary variables $\{\bU^{(s)}\}_{s=1}^{N_s}$, which may lead to increased computational complexity when $N_s$ is large. Another possible direction is to derive deterministic equivalents or closed-form approximations of the ergodic rate under specific channel distribution assumptions \cite{Zhang_Statistical}. In addition, robust design under imperfect CSI is also practically relevant, where the CSI uncertainty is commonly addressed through worst-case or outage-constrained QoS formulations  \cite{Zhou_robust, Hong_robust}. These extensions require dedicated analysis and are left for future work. 

\subsection{{\color{black}Low-Resolution Admittances}}
{\color{black}In previous discussions, we assumed that the tunable admittances in the reconfigurable impedance network have  infinite resolution such that the entries of the admittance matrix $\bB$ can take arbitrary continuous values (as long as they are nonzero). However, such an assumption incurs high hardware complexity. In practical implementations, the tunable admittances  are realized using finite-resolution components and therefore take values from a discrete set.

The proposed optimization framework can be extended to accommodate finite-resolution admittances by restricting the feasible set of each entry of $B_{i,j}$ to a predefined discrete alphabet, denoted by $\mathcal{D}_{i,j}$, for all $B_{i,j}\neq0$.   Under this setting, the overall algorithmic framework remains unchanged, and the only modification arises in the $\bB$-subproblem. Specifically, the $\bB$-subproblem becomes a discrete quadratic programming, which can be efficiently solved using standard discrete optimization techniques, such as exhaustive search (over small alphabets), branch-and-bound methods \cite{lawler1966branch}, or projected gradient method \cite{discretepgd}. A detailed investigation of the finite-resolution case is left for future work.}

\subsection{{\color{black}Other Modes}}
{\color{black}Till now, we have assumed a BD-RIS operating in the reflective mode, where the BS and the users are located on the same side of the BD-RIS. As discussed in \cite{group_conn}, a BD-RIS can also work in other modes, such as the transmissive mode, which can be realized using an $M$-cell BD-RIS composed of $2M$ back-to-back placed unidirectional antennas and is useful when the BS and the users are located on opposite sides of the surface.

The proposed algorithmic framework can also be applied to other operating modes by appropriately modifying the channel model. As an example, we consider a transmissive deployment in which the BS is positioned on the side of the first BD-RIS sector, while the users are located on the opposite side of the second sector. Accordingly, the channel between the BS and the BD-RIS is modeled as $\mathbf{G}=[\mathbf{G}_1^T,\mathbf{0}_{M\times N}^T]^T\in\C^{2M\times N}$, where $\mathbf{G}_1\in\C^{M\times N}$ is the channel from the BS to the first sector of the BD-RIS. Similarly, the channel between the BD-RIS and user $k$ is modeled as $\h_k=[\mathbf{0}_{M\times 1}^T,\h_{k,2}^T]^T\in\C^{2M\times 1}$, where $\h_{k,2}\in\C^{M\times 1}$ is the channel between the second layer of the BD-RIS and user $k$; see \cite{group_conn} for more details. Under this channel model, the proposed algorithm can be directly applied to the transmissive BD-RIS without any modification to its algorithmic structure.}
 \section{Simulation Results}\label{sec:5}
In this section, we present simulation results to demonstrate the effectiveness of the proposed algorithms and evaluate the performance of various BD-RIS architectures.

In our simulation, we adopt the same setup as in \cite{group_conn}, 
 where the distance between the BS and the BD-RIS is $d_{BI}=50$ m, and $K$ single-antenna users are randomly located near the BD-RIS with the same distance $d_{IU}=2.5$ m. Unless otherwise stated, the number of antennas at the BS is $N=4$, and the number of users is $K=4$. The channels from the BS to the BD-RIS and from the BD-RIS to the users are  modeled with both large-scale fading and small-scale fading. For large-scale fading, we use the distance-dependent path loss  $\zeta(d)=\zeta_0d^{-\alpha}$, where $\zeta_0$ represents the signal attenuation at a reference distance $d_0=1\,$m, which is set as $\zeta_0=-30$ dB,  and $\alpha$ is the path-loss exponent, which is set as $\alpha=2.2$ for both the channel from the BS to the BD-RIS and the channel from the BD-RIS to the users.  The small-scale fading is modeled as Rician fading with a  Rician factor of $\kappa=2$ dB for both channels. The noise power  is set as $\sigma^2=-80$ dBm.

\subsection{Comparison with Existing Algorithms}\label{sec:comparison_algorithm}

We first evaluate the effectiveness of the proposed pp-ADMM algorithms by comparing them with the state-of-the-art approaches. Since no existing algorithm is applicable to general BD-RIS architectures, we focus on fully- and group-connected BD-RIS architectures, where the group size is set as $G_s=4$. For the sum-rate maximization problem, we include the PDD algorithm in \cite{PDD} and the two-stage FP algorithm in \cite{twostage} as benchmarks.  For the transmit power minimization problem, we include the PDD algorithm in \cite{PDD} as a benchmark, which, to the best of our knowledge, is the only existing algorithm for solving the transmit power minimization problem in the  BD-RIS literature.

Figs. \ref{sumrate_algorithm} and \ref{sumrate_time} present the simulation results for the sum-rate maximization problem, where Fig. \ref{sumrate_algorithm} shows the sum-rate achieved by different algorithms, and Fig. \ref{sumrate_time} illustrates their CPU times, with the y-axis plotted on a logarithmic scale for better visualization. As shown in the figures, the proposed pp-ADMM algorithm achieves a similar sum-rate performance compared to PDD with a significantly reduced  CPU time.  
Both of these two approaches yield a substantially higher sum-rate than the two-stage FP algorithm in \cite{twostage}. The superiority is particularly prominent when the number of  RIS elements is large, in which cases a $20\%$ gain can be observed. Nevertheless, the two-stage FP algorithm is the most  computationally efficient, whose computational cost remains almost constant as $M$ grows. This is because in two-stage FP,  the scattering matrix $\bthe\in\C^{M\times M}$ is first determined in a heuristic way and then held fixed while solving the sum-rate maximization problem. Though simple and efficient, such a heuristic approach lacks theoretical justification and leads to performance loss. 

\begin{figure}[t]
\includegraphics[scale=0.3]{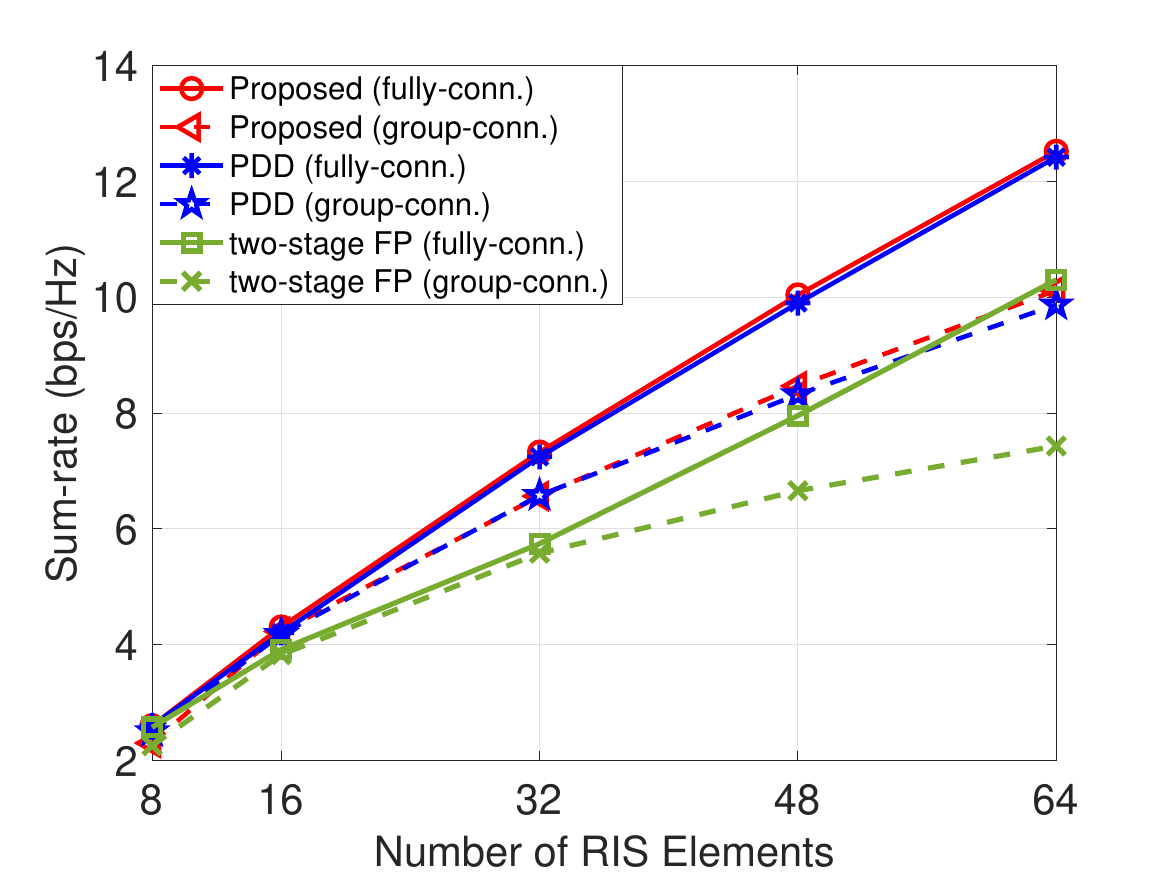}
\centering
\caption{Sum-rate versus number of RIS elements  ($N=K=4, G_s=4,$ $P_T=10$\,dBm).}
\label{sumrate_algorithm}
\end{figure}
\begin{figure}[t]
\includegraphics[scale=0.3]{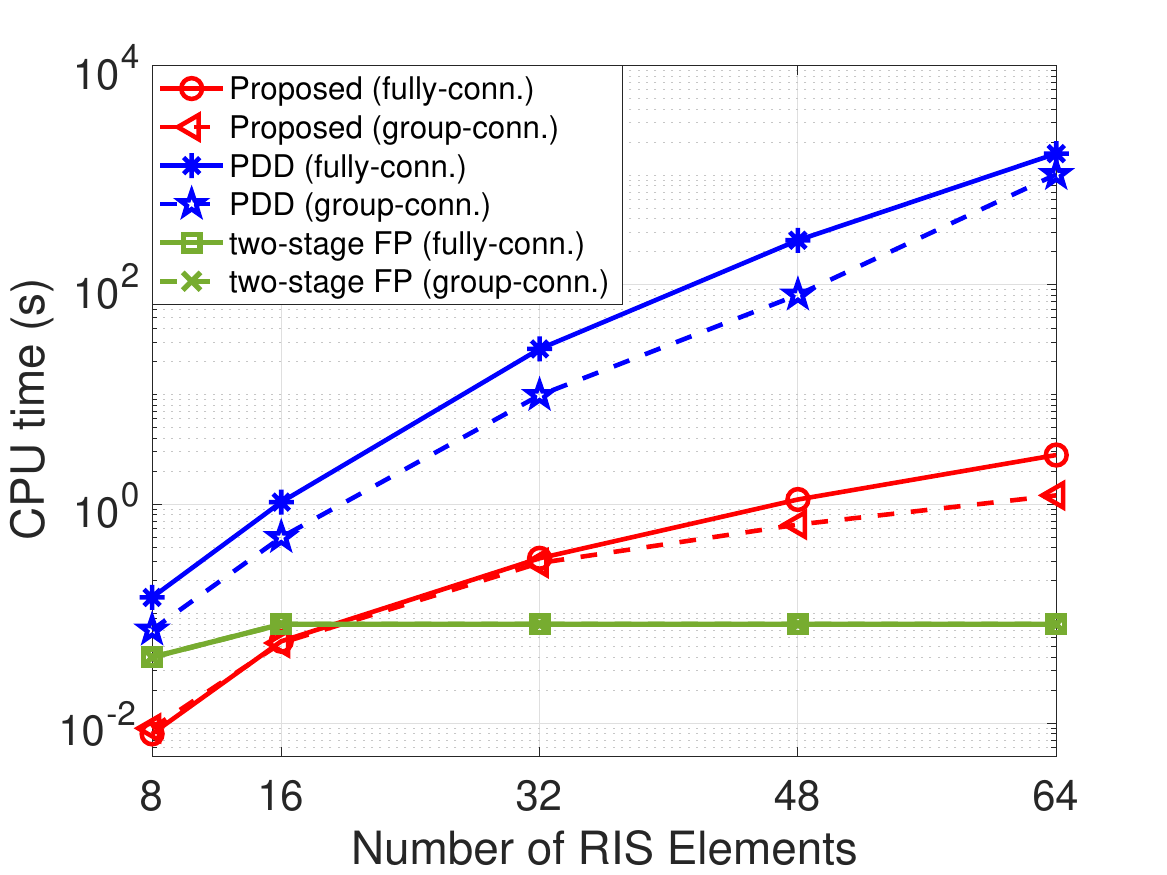}
\centering
\caption{CPU time versus number of RIS elements for sum-rate maximization ($N=K=4, G_s=4$, $P_T=10$\,dBm).}
\label{sumrate_time}
\end{figure}

In Figs. \ref{minpower} and \ref{minpower_time}, we evaluate the transmit power minimization problem,  comparing the performance and CPU time of the proposed pp-ADMM algorithm with the PDD algorithm in \cite{PDD}. The results show that the proposed pp-ADMM algorithm achieves lower transmit power than PDD with substantially reduced CPU time,  demonstrating both its effectiveness and  efficiency.  The computational advantage is even more pronounced than in the sum-rate maximization problem, since PDD needs to call CVX to solve a SOCP at each iteration, as discussed in Remark \ref{remark}, making it computationally inefficient even when the number of RIS elements is small.

We note that, for fully-connected and group-connected RIS, the proposed admittance-matrix formulation can be mathematically viewed as  parameterizing the symmetric and unitary constraints on the scattering matrix through the Cayley transform, which provides an equivalent formulation of the conventional models in  \cite{twostage,PDD}. The observed improvements in sum-rate or transmit power performance and CPU time reflect the effectiveness of the proposed pp-ADMM algorithms that exploit this equivalent reformulation. 

\begin{figure}[t]
\includegraphics[scale=0.3]{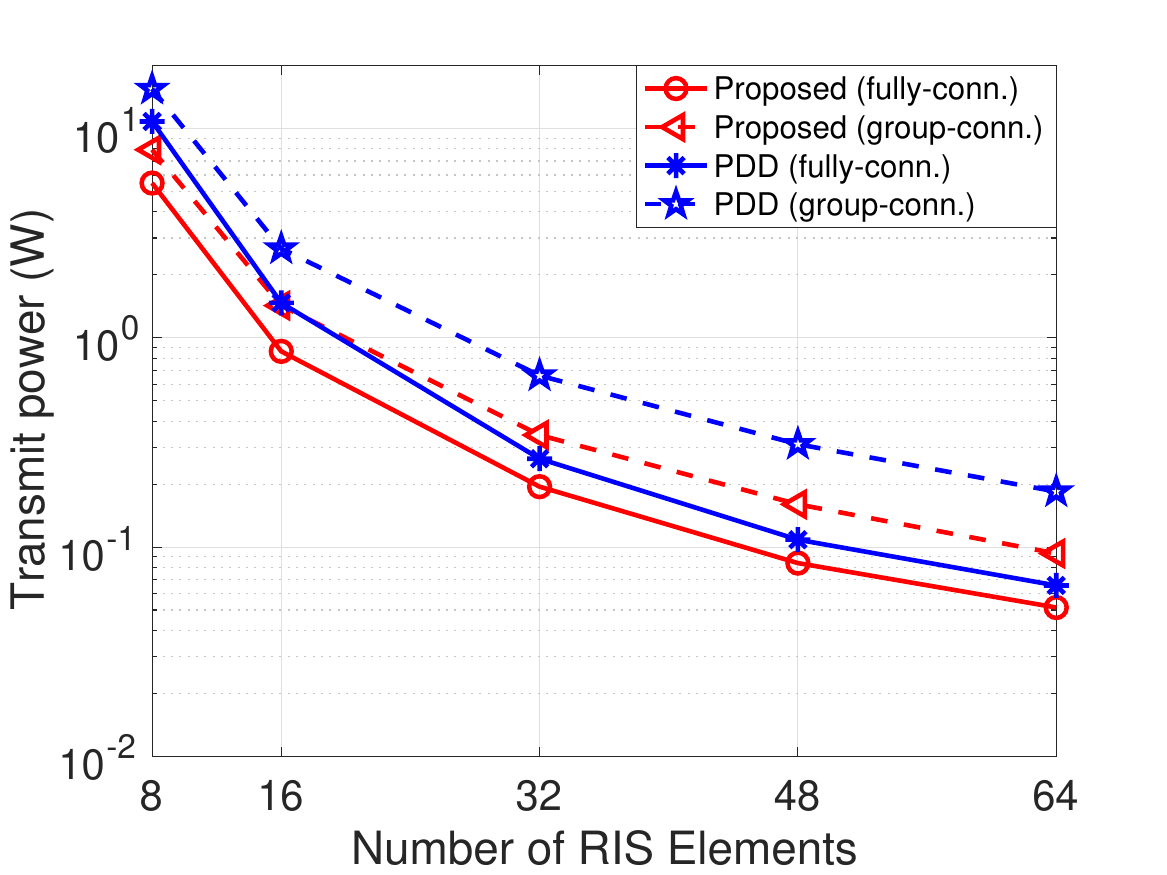}
\vspace{-0.2cm}
\centering
\caption{Transmit power versus number of RIS elements ($N=K=4, G_s=4${\color{black}, $\Gamma_k=10$ dB}).}
\label{minpower}
\end{figure}

\begin{figure}[t]
\includegraphics[scale=0.3]{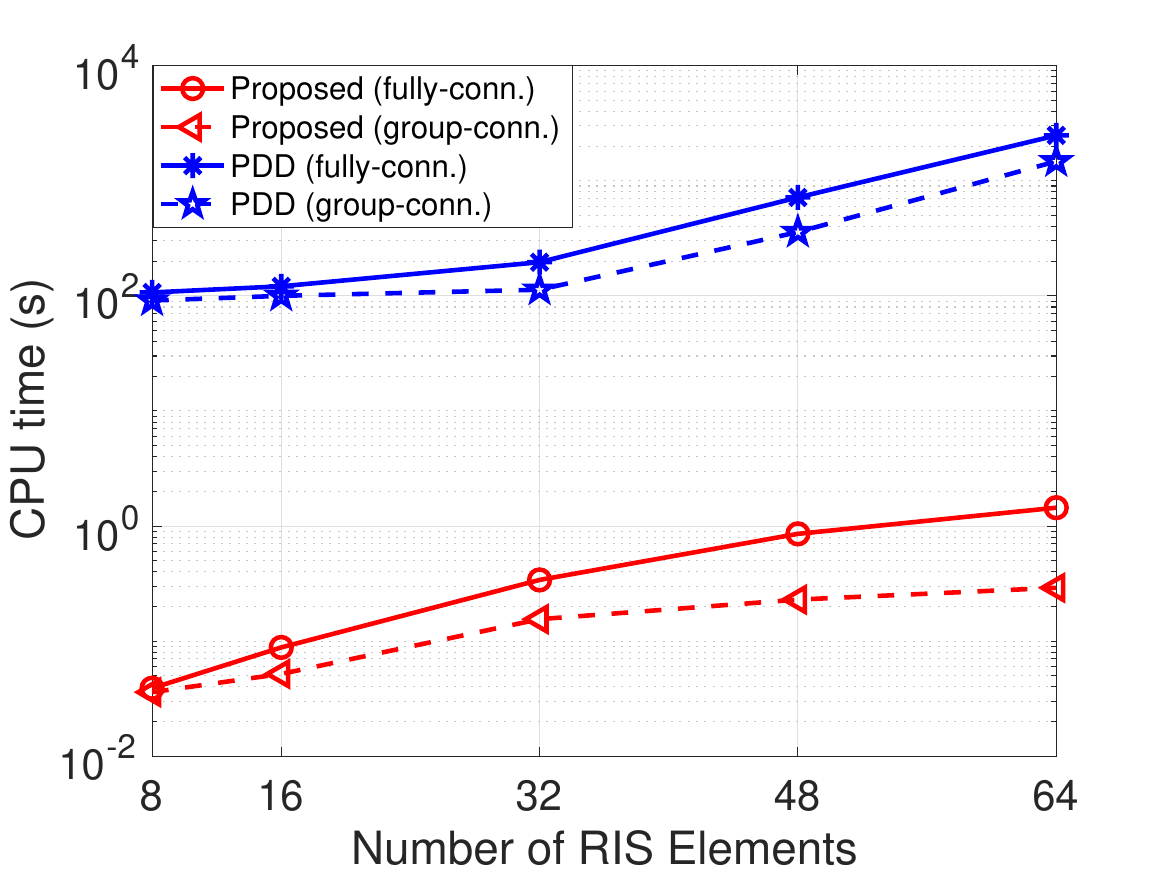}
\vspace{-0.2cm}
\centering
\caption{CPU time versus number of RIS elements for transmit power minimization  ($N=K=4, G_s=4${\color{black}, $\Gamma_k=10\,$dB}).}
\label{minpower_time}
\end{figure}

\begin{figure}[t]
\includegraphics[scale=0.3]{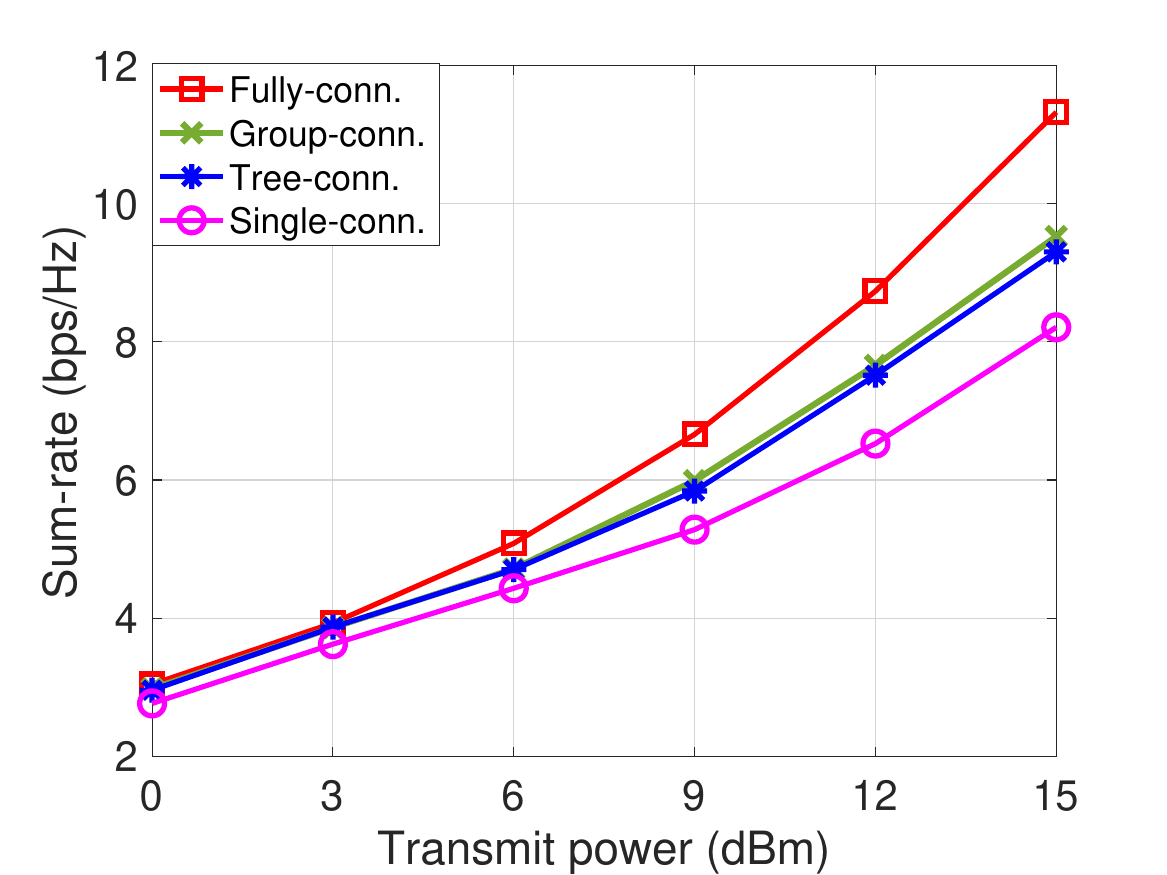}
\centering
\vspace{-0.2cm}
\caption{Sum-rate versus transmit power for different BD-RIS architectures ($N=K=4, M=32, G_s=4$).}
\label{sumrate_structure}
\end{figure}

\begin{figure}[t]
\includegraphics[scale=0.3]{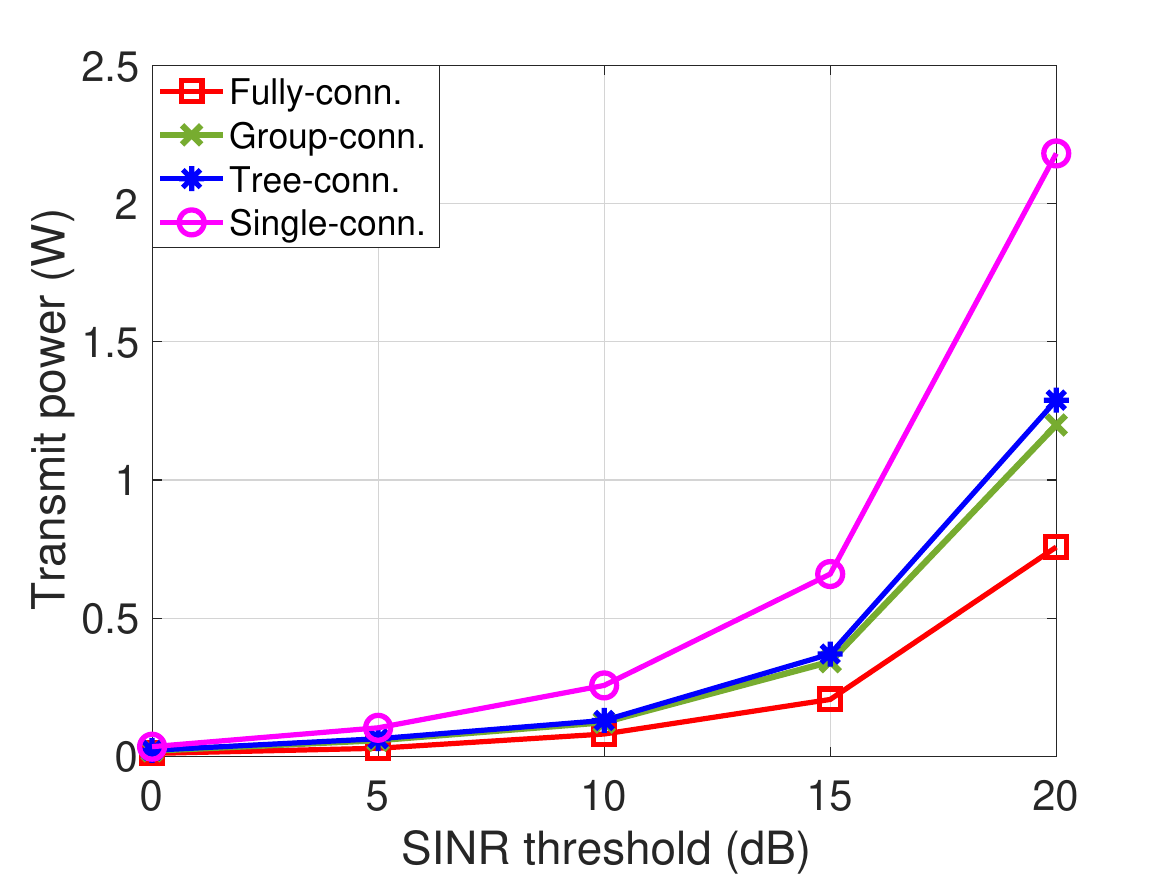}
\centering
\vspace{-0.2cm}
\caption{Transmit power versus SINR threshold for different BD-RIS architectures ($N=K=4, M=32, G_s=4$).}
\label{minpower_structure}
\end{figure}
\subsection{Performance of Different BD-RIS Architectures}
In Figs. \ref{sumrate_structure} and \ref{minpower_structure}, we compare the performance of the four BD-RIS architectures introduced in Section \ref{sec:RIS_arch}, where Fig. \ref{sumrate_structure} presents the result for sum-rate maximization and Fig. \ref{minpower_structure} shows the result for transmit power minimization. As expected, the fully-connected and single-connected RIS demonstrate the best and worst performance, respectively, due to the highest and lowest design flexibility they offer.  The tree-connected RIS, though proven optimal for single-user MISO systems \cite{tree},  suffers from a performance degradation compared to fully-connected RIS in   MU-MISO systems. Nevertheless, we can observe that the tree-connected RIS achieves performance comparable to the group-connected RIS with group size $G_s = 4$, while requiring fewer impedances (the impedances are $2{M}-1$ and $2.5{M}$, respectively, for tree-connected and group-connected RIS with $G_s=4$). This suggests that the tree-connected architecture may still be preferable to the group-connected ones in multiuser systems, possibly due to its ``more connected'' circuit topology.

\section{Conclusion}\label{sec:6}
This paper studied the optimization of BD-RIS, focusing mainly on sum-rate maximization and transmit power minimization for MU-MISO systems. For each problem, we custom-developed   a  pp-ADMM algorithm. The proposed algorithms offer several advantages. First, they are architecture-independent and thus applicable to any BD-RIS architecture. Second, they are computationally efficient, with each subproblem either admitting a closed-form solution or being solvable with low complexity.  Finally, they exhibit desirable theoretical convergence properties under mild assumptions. The proposed framework can be generalized in various aspects, including accommodating other utility functions, extending to MU-MIMO systems, incorporating statistical CSI, supporting different operating modes, and considering finite-resolution admittances.  Extensive simulation results demonstrate that our approaches achieve better trade-off between performance and CPU time compared to existing state-of-the-art methods. 

Another interesting insight drawn from our simulations is that the tree-connected RIS, known to be optimal in single-user MISO systems, is no longer optimal in multiuser systems. {\color{black}Motivated by this observation, we theoretically identify the optimal BD-RIS architectures in multiuser systems in our recent work \cite{graphtit}.  The current paper has served as an algorithmic tool in \cite{graphtit} to validate its theoretical results.}

\appendices
\section{Proof of \eqref{rjstar}}\label{app:rjstar}
In this appendix, we prove that \eqref{rjstar} is  the optimal solution to  problem \eqref{y-problem2}.

We first examine the KKT condition of problem \eqref{y-problem2}. Let $\eta$ and $\nu$ be the Lagrange multipliers associated with the first and second constraints of problem \eqref{y-problem2}, respectively. The KKT condition is given by 
\begin{subequations}\label{kkt}
\begin{align}
&\left\{
\begin{aligned}
(1-\eta)r_{k,k}=a_{k,k}+\nu/2,~~~~&\text{if }j=k;\\
(1+\Gamma_k\eta)r_{k,j}=a_{k,j},~\,~~~~~~~~&\text{if }j\neq k,
\end{aligned}\right.\label{rkj}\\
&\eta\left(r_{k,k}^2-\Gamma_k(\sum_{j\neq k}r_{k,j}^2+\sigma^2)\right)=0,\label{complementary}\\
&r_{k,k}^2\geq\Gamma_k(\sum_{j\neq k}r_{k,j}^2+\sigma^2), ~~\eta\geq 0,\label{constraint1}\\
&\nu r_{k,k}=0,~\nu\geq 0,~r_{k,k}\geq 0.\label{constraint2}
\end{align}
\end{subequations}
It follows  from \eqref{constraint1} and \eqref{constraint2} that $r_{k,k}>0$ and $\nu=0$. 
This, combined with the expression of $r_{k,k}$ in \eqref{rkj}, implies that $\eta\in[0,1)$ when $a_{k,k}> 0$, $\eta=1$ when $a_{k,k}=0$,  and $\eta\in(1,\infty)$ otherwise.  Clearly, when $a_{k,k}=0$, the solution given by \eqref{rjstar} is the unique KKT point satisfying  \eqref{rkj}-\eqref{constraint2}.  Next, we investigate the case that $a_{k,k}\neq 0$. Consider the following function:
$$
\begin{aligned}
h(\eta)&=r_{k,k}^2-\Gamma_k(\sum_{j\neq k}r_{k,j}^2+\sigma^2)\\
&=\frac{a_{k,k}^2}{(1-\eta)^2}-\frac{\Gamma_k\sum_{j\neq k}a_{k,j}^2}{(1+\Gamma_k\eta)^2}-\Gamma_k\sigma^2,
\end{aligned}$$
where the second equation is due to \eqref{rkj} and the fact that $\nu=0$. Then \eqref{complementary} and \eqref{constraint1} can be expressed as $$\eta h(\eta)=0, ~\eta\geq 0, ~h(\eta)\geq 0.$$ It is straightforward to check that 
\begin{equation}\label{hproperty}
\begin{aligned}
h(0)&={a_{k,k}^2}-{\Gamma_k\sum_{j\neq k}a_{k,j}^2}-\Gamma_k\sigma^2,\\
\lim_{\eta\to 1}h(\eta)&=\infty,~~\lim_{\eta\to \infty}h(\eta)=-\Gamma_k\sigma^2<0.
\end{aligned}\end{equation}
Therefore, $\eta^*$ given in \eqref{eta2} is well-defined (i.e., $\eta_1^*$ and $\eta_2^*$ exist) and the solution given by \eqref{rjstar} satisfies the KKT condition in \eqref{rkj}-\eqref{constraint2}.  In addition, according to
$$h'(\eta)=2\Gamma_k^2\sum_{j\neq k}a_{k,j}^2(1+\Gamma_k\eta)^{-3}+2a_{k,k}^2(1-\eta)^{-3},$$
 $h(\eta)$ is increasing in $(0,1)$. Since there is at most one solution to $h'(\eta)=0$ in $(1,\infty)$, $h(\eta)$ is either decreasing or first decreasing and then increasing in $(1,\infty)$. This, combined with \eqref{hproperty}, implies the uniqueness of   $\eta_1^*$ and $\eta_2^*$. Therefore,  the solution given by \eqref{rjstar} is the unique KKT point of \eqref{y-problem2}, which is thus the optimal solution to \eqref{y-problem2}. 

\section{Proof of Theorem 1}\label{appendix:converge}\vspace{-0.1cm}
 Our proof follows a similar idea to the classical proof in \cite{ADMM,ADMM2}, which contains three main steps.  First, we show that the difference between the Lagrangian multipliers in successive iterates can be bounded by that between the primal variables. With this result, the second step establishes a sufficient decrease condition for the augmented Lagrangian function. Finally, we investigate the limiting behavior of the sequence generated by the proposed pp-ADMM algorithm, based on the conclusions from the first two steps and the update rules in (14a) -- (14f).

Next, we first present the results for the first and second steps  as two auxiliary lemmas in Appendix \ref{app:AA}, and then conduct the final step to prove Theorem 1 in Appendix \ref{app:AB}.
\vspace{-0.3cm}
\subsection{Auxiliary Lemmas}\label{app:AA}
\begin{lemma}\label{lambdabound}
Under the assumption in Theorem 1, there exists a constant $C_0$ independent of $\rho$ such that \vspace{-0.2cm}

{\small$$
\begin{aligned}\left\|\blam^{t+1}-\blam^t\right\|^2_2\leq C_0&\big(\|\y^{t+1}-\y^t\|_2^2+\|\boldsymbol{\gamma}^{t+1}-\boldsymbol{\gamma}^t\|_2^2\\
&\hspace{-2.7cm}+\|\bB^{t+1}-\bB^t\|_F^2+\|\bW^{t+1}-\bW^t\|^2_F+\left\|\bU^{t+1}-\bU^t\right\|^2_F\big).
\end{aligned}$$}
\end{lemma}
\begin{proof}
According to the update rule of variable $\bU$ in (14e), we have \vspace{-0.2cm}

{\small$$
\nabla_{\bU}\mathcal{L}_{\rho}(\y^{t+1},\boldsymbol{\gamma}^{t+1},\bW^{t+1},\bB^{t+1},\bU^{t+1},\blam^t)=\mathbf{0}, $$}which implies that\vspace{-0.15cm}
 
{\small$$
\begin{aligned}
&\nabla_{\bU}{\tilde{R}}(\y^{t+1},\boldsymbol{\gamma}^{t+1},\bW^{t+1},\bU^{t+1})\\
&\hspace{-0.05cm}-\hspace{-0.05cm}(\mathbf{I}\hspace{-0.05cm}+\hspace{-0.05cm}\mathrm{i}Z_0\bB^{t+1})\hspace{-0.05cm}\left(\blam^t\hspace{-0.05cm}+\hspace{-0.05cm}\rho((\mathbf{I}\hspace{-0.04cm}-\hspace{-0.04cm}\mathrm{i}Z_0\bB^{t+1})\bU^{t+1}\hspace{-0.1cm}-\hspace{-0.05cm}(\mathbf{I}+\hspace{-0.05cm}\mathrm{i}Z_0\bB^{t+1})\H)\right)\\
=\,&\nabla_{\bU}{\tilde{R}}(\y^{t+1},\boldsymbol{\gamma}^{t+1},\bW^{t+1},\bU^{t+1})-(\mathbf{I}+\mathrm{i}Z_0\bB^{t+1})\blam^{t+1}=\mathbf{0},\end{aligned}
$$}
\hspace{-0.15cm}where the first equality uses the update rule of the Lagrange multiplier $\blam$ in (14f). We claim that for any $t\geq 0$,  the matrix $\mathbf{I}+\mathrm{i}Z_0\bB^{t}$ is invertible with  \vspace{-0.1cm}

{\small \begin{equation}\label{property_B}
\|(\mathbf{I}+\mathrm{i}Z_0\bB^{t})^{-1}\|_2\leq 1.
\end{equation}}
\hspace{-0.2cm}This can be seen by expressing the matrix  $\mathbf{I}+\mathrm{i}Z_0\bB^{t}$ as \vspace{-0.2cm}

{\small$$\mathbf{I}+\mathrm{i}Z_0\bB^{t}=\mathbf{V}(\mathrm{i}Z_0\bD+\mathbf{I})\mathbf{V}^T,$$}
\hspace{-0.13cm}where $\bB^{t}=\mathbf{V}\bD\mathbf{V}^T$ is the eigenvalue decomposition of $\bB^{t}$. Hence,  $\blam^{t+1}$ can be expressed as    \vspace{-0.2cm}

{\small\begin{equation}\label{lambda}
\hspace{-0.15cm}\blam^{t+1}\hspace{-0.05cm}=\hspace{-0.05cm}(\mathbf{I}\hspace{-0.01cm}+\hspace{-0.01cm}\mathrm{i}Z_0\bB^{t+1})\hspace{-0.05cm}^{-1}\nabla_{\bU}{\tilde{R}}(\y^{t+1}\hspace{-0.05cm},\boldsymbol{\gamma}^{t+1},\bW^{t+1},\bU^{t+1}).\end{equation}}
\hspace{-0.2cm} 
Let
{\small$\nabla_{\bU}\tilde{R}^t:=\nabla_{\bU}\tilde{R}(\y^{t},\boldsymbol{\gamma}^t,\bW^t,\bU^t).$} It follows from \eqref{lambda} that 
{\small\begin{equation}\label{lambda:1}
\begin{aligned}
&\|\blam^{t+1}\hspace{-0.05cm}-\blam^t\|^2_F\\
\leq\,& 2\left\|(\mathbf{I}\hspace{-0.05cm}+\hspace{-0.05cm}\mathrm{i}Z_0\bB^{t+1})^{-1}\left(\nabla_{\bU}{\tilde{R}}^{t+1}\hspace{-0.05cm}-\hspace{-0.05cm}\nabla_{\bU}{\tilde{R}}^t\right)\right\|_F^2\\
 &+2\left\|\left((\mathbf{I}+\mathrm{i}Z_0\bB^{t+1})^{-1}-(\mathbf{I}+\mathrm{i}Z_0\bB^t)^{-1}\right)\nabla_{\bU} \tilde{R}^t\right\|_F^2\\
 \leq \,&2\left\|\nabla_{\bU}{\tilde{R}}^{t+1}\hspace{-0.05cm}-\hspace{-0.05cm}\nabla_{\bU}{\tilde{R}}^t\right\|_F^2+2Z_0^2\|\nabla_\bU \tilde{R}^t\|
_2^2\|\bB^{t+1}-\bB^t\|_F^2,
\end{aligned}
\end{equation}}
\hspace{-0.15cm}where the last inequality uses \eqref{property_B}, the fact that  $\|\mathbf{X}\mathbf{Y}\|_F\leq \|\mathbf{X}\|_2\|\mathbf{Y}\|_F$ for any given matrices $\mathbf{X}$ and $\mathbf{Y}$, and \vspace{-0.2cm}

{\small$$
\begin{aligned}
&(\mathbf{I}+\mathrm{i}Z_0\bB^{t+1})^{-1}-(\mathbf{I}+\mathrm{i}Z_0\bB^t)^{-1}\\
=\,&\mathrm{i}Z_0 (\mathbf{I}+\mathrm{i}Z_0\bB^{t+1})^{-1}(\bB^{t}-\bB^{t+1})(\mathbf{I}+\mathrm{i}Z_0\bB^{t})^{-1}.
\end{aligned}$$}
\hspace{-0.1cm}Note that $\{\bW^t\}_{t\geq 0}$ is bounded due to the total transmit power constraint. Under the assumption in Theorem 1, $\{\bU^t\}_{t\geq 0}$ is  bounded. Then it follows from (15) and (16) that $ \{\y^t\}_{t\geq 0}$ and $\{\boldsymbol{\gamma}^t\}_{t\geq0}$ are bounded.  Since $\nabla_\bU \tilde{R}^t$ is continuously differentiable, it is bounded and Lipschitz continuous on any compact set, i.e., there exists  constants $C_{0,1}$ and $C_{0,2}$ such that $\|\nabla_\bU \tilde{R}^t\|_2^2\leq C_{0,1}$ and \vspace{-0.2cm}

{\small$$\begin{aligned}\left\|\nabla_{\bU}\tilde{R}^{t+1}\hspace{-0.07cm}-\hspace{-0.07cm}\nabla_{\bU}\tilde{R}^t\right\|_F^2&\hspace{-0.05cm}\leq C_{0,2}\hspace{-0.05cm}\left(\|\y^{t+1}-\y^t\|_2^2+\|\boldsymbol{\gamma}^{t+1}-\boldsymbol{\gamma}^t\|_2^2\right.\\
&\left.+\|\bW^{t+1}\hspace{-0.05cm}-\hspace{-0.05cm}\bW^t\|^2_F+\left\|\bU^{t+1}\hspace{-0.05cm}-\hspace{-0.05cm}\bU^t\right\|^2_F\hspace{-0.05cm}\right).
\end{aligned}$$}
\hspace{-0.15cm}Combining the above discussions with \eqref{lambda:1}, we can conclude that Lemma \ref{lambdabound} holds with $C_0=2\max\{Z_0^2C_{0,1},C_{0,2}\}$.
\end{proof}
\begin{lemma}\label{ALbound}
Let 
$\mathcal{L}_{\rho}^t:=\mathcal{L}_{\rho}(\y^{t},\boldsymbol{\gamma}^{t},\bW^{t},\bB^{t},\bU^t,\blam^t)$ denote the value of the augmented  Lagrangian function in \text{\normalfont{(13)}} at the $t$-th iteration. Then the following inequality holds: \vspace{-0.2cm}

{\small\begin{equation}\label{ALMdiff}
\begin{aligned}
&\mathcal{L}_{\rho}^{t+1}-\mathcal{L}_{\rho}^t\\
\geq&\left(\frac{\sigma^2}{2}\hspace{-0.05cm}-\hspace{-0.05cm}\frac{C_0}{\rho}\right)\|\y^{t+1}\hspace{-0.05cm}-\hspace{-0.05cm}\y^t\|_2^2+\hspace{-0.05cm}\left(\frac{1}{2(1+C_{\gamma})^2}\hspace{-0.05cm}-\hspace{-0.05cm}\frac{C_0}{\rho}\right)\|\boldsymbol{\gamma}^{t+1}\hspace{-0.05cm}-\hspace{-0.05cm}\boldsymbol{\gamma}^t\|_2^2\\
&+\left(\frac{\tau}{2}-\frac{C_0}{\rho}\right)\|\bW^{t+1}-\bW^t\|^2_F+\left(\frac{\xi}{4}-\frac{C_0}{\rho}\right)\|\bB^{t+1}-\bB^t\|_F^2\\&+\left(\frac{\rho}{2}-\frac{C_0}{\rho}\right)\left\|\bU^{t+1}-\bU^t\right\|^2_F,
\end{aligned}
\end{equation}}

\hspace{-0.35cm}where $C_\gamma$ is an upper bound\footnote{$C_{\gamma}$ exists due to the boundedness of $\{\bW^t\}$, $\{\bU^t\}$; see discussions in the proof of Lemma \ref{lambdabound}.}  on $\|\boldsymbol{\gamma}^t\|_{\infty}$, $C_0$ is given in Lemma \ref{lambdabound}, $\tau$ and  $\rho$ are algorithm parameters; see \text{\normalfont{(14)}}.  
\end{lemma}

\begin{proof}
The proof is based on the fact that given a strongly concave function $f(\x)$ with modulus $\mu$ and a convex set $\mathcal{C}$, the following inequality holds for $\x^*\in\arg\max_{\x\in\mathcal{C}}~f(\x)$: 
$$f(\x^*)-f(\x)\geq \frac{\mu}{2}\|\x-\x^*\|^2_2,~~\forall~\x\in\mathcal{C}.$$ 
When the variable is a matrix, the $\ell_2$ norm on the r.h.s. should be replaced by the Frobenius norm. 
%
%
Applying this result to the $\y$-, $\boldsymbol{\gamma}$-, $\bW$-, $\bB$-, and $ \bU$-subproblems in (14a)--(14e), whose objective functions are strongly concave with modulus ${\sigma^2}$, $(1+C_{\gamma})^{-2}$,  $\tau$, $\frac{\xi}{2}$, and $\rho$, respectively,  we get\vspace{-0.3cm}

{\small\begin{equation}\label{decrease:4}
\hspace{-0.2cm}
\begin{aligned}
&\mathcal{L}_{\rho}(\y^{t+1},\boldsymbol{\gamma}^{t+1},\bW^{t+1},\bB^{t+1},\bU^{t+1},\blam^{t})\\
&-\mathcal{L}_{\rho}(\y^{t},\boldsymbol{\gamma}^{t},\bW^{t},\bB^{t},\bU^t,\blam^t)\\
\geq&\frac{\sigma^2}{2}\|\y^{t+1}-\y^t\|_2^2+\frac{1}{2(1+C_{\gamma})^2}\|\boldsymbol{\gamma}^{t+1}-\boldsymbol{\gamma}^t\|^2_2\\
&\hspace{-0.2cm}+\frac{\tau}{2}\|\bW^{t+1}\hspace{-0.05cm}-\hspace{-0.05cm}\bW^t\|^2_F+ \frac{\xi}{4}\|\bB^{t+1}\hspace{-0.05cm}-\hspace{-0.05cm}\bB^t\|^2_F+ \frac{\rho}{2}\|\bU^{t+1}-\hspace{-0.05cm}\bU^t\|^2_F.
\end{aligned}
\end{equation}}
\hspace{-0.08cm}In addition,  using the formula of the augmented Lagrangian function in (13), we have \vspace{-0.08cm}
{\small\begin{equation}\label{decrease:6}
\begin{aligned}
&\mathcal{L}_{\rho}(\y^{t+1},\boldsymbol{\gamma}^{t+1},\bW^{t+1},\bB^{t+1},\bU^{t+1},\blam^{t+1})\\
&-\mathcal{L}_{\rho}(\y^{t+1},\boldsymbol{\gamma}^{t+1},\bW^{t+1},\bB^{t+1},\bU^{t+1},\blam^t)\\
=\,&-\left<\blam^{t+1}\hspace{-0.05cm}-\blam^t,(\mathbf{I}-\mathrm{i}Z_0\bB^{t+1})\bU^{t+1}-(\mathbf{I}+\mathrm{i}Z_0\bB^{t+1})\bH\right>\\
=\,&-\frac{1}{\rho}\|\blam^{t+1}-\blam^t\|_F^2\\
\end{aligned}\vspace{-0.2cm} 
\end{equation}}
\hspace{-0.18cm}where the second equality uses the update rule of $\blam$ in (14f). Combining  \eqref{decrease:4} and \eqref{decrease:6} and applying Lemma \ref{lambdabound} to \eqref{decrease:6} gives the desired result in Lemma \ref{ALbound}.
\end{proof}
\subsection{Proof of Theorem 1}\label{app:AB}
Now we are ready to prove Theorem 1.
 Given $T>0$, summing over \eqref{ALMdiff} from $t=0$ to $t=T$ gives
{\small\begin{equation}\label{ALMdiff2}
\begin{aligned}
&\mathcal{L}_{\rho}^{T+1}-\mathcal{L}_{\rho}^0\\
\geq \hspace{-0.05cm}&\left(\frac{\sigma^2}{2}-\frac{C_0}{\rho}\right)\sum_{t=0}^{T}\|\y^{t+1}-\y^t\|_2^2+\hspace{-0.05cm}\left(\frac{\xi}{4}-\frac{C_0}{\rho}\right)\sum_{t=1}^T\|\bB^{t+1}\hspace{-0.05cm}-\hspace{-0.05cm}\bB^t\|_F^2\\
&\hspace{-0.05cm}+\hspace{-0.05cm}\left(\hspace{-0.05cm}\frac{1}{2(1+C_{\gamma})^2}-\frac{C_0}{\rho}\hspace{-0.05cm}\right)\sum_{t=0}^{T}\|\boldsymbol{\gamma}^{t+1}\hspace{-0.05cm}-\hspace{-0.05cm}\boldsymbol{\gamma}^t\|_2^2\\
&\hspace{-0.05cm}+\hspace{-0.05cm}\left(\frac{\tau}{2}\hspace{-0.05cm}-\hspace{-0.05cm}\frac{C_0}{\rho}\right)\hspace{-0.05cm}\sum_{t=0}^{T}\|\bW^{t+1}\hspace{-0.08cm}-\hspace{-0.08cm}\bW^t\|^2_F\hspace{-0.05cm}+\hspace{-0.05cm}\left(\frac{\rho}{2}\hspace{-0.05cm}-\hspace{-0.05cm}\frac{C_0}{\rho}\right)\hspace{-0.05cm}\sum_{t=0}^{T}\left\|\bU^{t+1}\hspace{-0.08cm}-\hspace{-0.08cm}\bU^t\right\|^2{\hspace{-0.2cm}_F}.
\end{aligned}
\end{equation}}
\hspace{-0.2cm}Let
{\small$$
\rho_0=\max\left\{\frac{2C_0}{\sigma^2}, \frac{4C_0}{\xi}, 2(1+C_\gamma)^2C_0,\frac{2C_0}{\tau},\sqrt{2C_0}\right\}.
$$}
\hspace{-0.2cm}It is straightforward to check  that if  $\rho>\rho_0$, the factor in front of each term in \eqref{ALMdiff2} is positive. 
In addition, $\{\mathcal{L}_{\rho}^t\}_{t\geq 0}$ is bounded from above under boundedness assumption on $\{\bU^t\}_{t\geq 0}$ and $\{\blam^t\}_{t\geq 0}$. 
Therefore,  letting $T\to\infty$ in \eqref{ALMdiff2} yields\vspace{-0.2cm}

{\small\begin{equation*}\label{to0}
\begin{aligned}
&\|\y^{t+1}-\y^t\|_2\to0,~\|\boldsymbol{\gamma}^{t+1}-\boldsymbol{\gamma}^t\|_2\to 0,~\|\bW^{t+1}-\bW^t\|_F\to 0,\\
&\|\bB^{t+1}-\bB^t\|_F\to0,~\|\bU^{t+1}-\bU^t\|_F\to 0.
\end{aligned}
\end{equation*}}
\hspace{-0.13cm}It further follows from Lemma \ref{lambdabound} that 
$\|\blam^{t+1}-\blam^t\|_F\to 0.$
According to the update rule of the proposed pp-ADMM algorithm in (14), we have

{\small\begin{subequations}\label{iterate}
\begin{align}
&\nabla_{\y}\tilde{R}(\y^{t+1},\boldsymbol{\gamma}^t,\bW^t,\bU^t)=\mathbf{0},\\
&\nabla_{\boldsymbol{\gamma}}\tilde{R}_(\y^{t+1},\boldsymbol{\gamma}^{t+1},\bW^t,\bU^t)=\mathbf{0},\\
&\hspace{-0.05cm}\big<\nabla_\bW \tilde{R}(\y^{t+1},\boldsymbol{\gamma}^{t+1},\bW^{t+1},\bU^t)\notag\\
&~-\tau(\bW^{t+1}\hspace{-0.05cm}-\hspace{-0.05cm}\bW^t),\bW\hspace{-0.05cm}-\hspace{-0.05cm}\bW^{t+1}\big>\hspace{-0.05cm}\leq 0,~\forall\,\bW\hspace{-0.1cm}:\|\bW\|_F^2\hspace{-0.05cm}\leq\hspace{-0.05cm} P_T,\\
&\big<\left(\blam^{t+1}-\rho(\mathbf{I}-\mathrm{i}Z_0\bB^{t+1})(\bU^{t+1}-\bU^t)\right)(\mathrm{i}Z_0\bU^t+\mathrm{i}Z_0\bH)^{\dagger}\notag\\&~-\frac{\xi}{2}(\bB^{t+1}\hspace{-0.1cm}-\hspace{-0.05cm}\bB^t)-\frac{\xi}{2}\text{diag}(\bB^{t+1}-\bB^t),\bB\hspace{-0.05cm}-\hspace{-0.05cm}\bB^{t+1}\big>\le 0,\notag\\
&\hspace{4cm}~\forall\,\bB: \bB=\bB^T,~\bB\in\mathcal{B},\\
&\nabla_{\bU} \tilde{R}(\y^{t+1}\hspace{-0.05cm},\boldsymbol{\gamma}^{t+1}\hspace{-0.05cm},\bW^{t+1}\hspace{-0.05cm},\bU^{t+1}\hspace{-0.03cm})\hspace{-0.05cm}-\hspace{-0.05cm}(\mathbf{I}\hspace{-0.03cm}+\hspace{-0.03cm}\mathrm{i}Z_0\bB^{t+1})\blam^{t+1}\hspace{-0.08cm}=\mathbf{0},\\
&\blam^{t+1}=\blam^t+\rho((\mathbf{I}-\mathrm{i}Z_0\bB^{t+1})\bU^{t+1}-(\mathbf{I}+\mathrm{i}Z_0\bB^{t+1})\bH),
\end{align}
\end{subequations}}
\hspace{-0.25cm}Let {\small$\left(\y^*,\boldsymbol{\gamma}^*,\bW^*,\bB^{*},\bU^*,\blam^*\right)$} be a limit point of {\small$\left(\y^t,\boldsymbol{\gamma}^t,\bW^t,\bB^{t},\bU^t,\blam^t\right)$} with\vspace{-0.3cm}

{\small$$\lim_{j\to\infty}\left(\y^{t_j},\boldsymbol{\gamma}^{t_j},\bW^{t_j},\bB^{t_j},\bU^{t_j},\blam^{t_j}\right)=\left(\y^*,\boldsymbol{\gamma}^*,\bW^*,\bB^{*},\bU^*,\blam^*\right).$$}
Then 
substituting $t=t_j$ in \eqref{iterate} and letting $j\to\infty$ gives 
{\small\begin{subequations}\label{iterate2}
\begin{align}
&\nabla_{\y}\tilde{R}(\y^{*},\boldsymbol{\gamma}^*,\bW^*,\bU^*)=\mathbf{0},\\
&\nabla_{\boldsymbol{\gamma}}\tilde{R}(\y^{*},\boldsymbol{\gamma}^{*},\bW^*,\bU^*)=\mathbf{0},\\
&\left<\nabla_\bW \tilde{R}(\y^{*},\boldsymbol{\gamma}^{*},\bW^{*},\bU^*),\bW-\bW^*\right>\leq{0},\notag\\
&\hspace{4cm}\forall~\bW:\|\bW\|_F^2\leq P_T,\\
&\left<\blam^{*}(\mathrm{i}Z_0\bU^*+\mathrm{i}Z_0\bH)^{\dagger},\bB-\bB^*\right>\le\mathbf{0},~\forall~\bB: \bB=\bB^T,~\bB\in\mathcal{B},\\
&\nabla_{\bU} \tilde{R}(\y^{*},\boldsymbol{\gamma}^{*},\bW^{*},\bU^{*})-(\mathbf{I}+\mathrm{i}Z_0\bB^*)\blam^{*}=\mathbf{0},\\
&(\mathbf{I}-\mathrm{i}Z_0\bB^{*})\bU^{*}-(\mathbf{I}+\mathrm{i}Z_0\bB^{*})\bH=\mathbf{0},
\end{align}
\end{subequations}}
\hspace{-0.1cm}i.e., $\left(\y^*,\boldsymbol{\gamma}^*,\bW^*,\bB^{*},\bU^*\right)$ is a stationary point of problem (12), with $\blam^*$ being the Lagrangian multiplier associated with the bilinear constraint (10b).

\section{Proof of Theorem 2}

The proof of Theorem 2 follows the same procedure as that of Theorem 1. The main difference is that there are two bilinear constraints  in problem (26), which couple the three variables $(\bW,\bU,\bB)$ and complicate the relationship between the primal and dual variables. As a result, the proof is slightly more involved. 

Following the same structure as the proof of Theorem 1, we begin by presenting two auxiliary lemmas, which, respectively, bound the difference between successive iterations of dual variables with that of primal variables,  and demonstrate the sufficient decrease property of a properly defined potential function. Then, we give   the proof of Theorem 2.

\setcounter{lemma}{2}
\begin{lemma}
Under the assumptions in Theorem 2, there exists constants $C_0>0$ and $C_1>0$ such that \vspace{-0.2cm}

{\small$$\begin{aligned}
\,\|\blam^{t+1}-\blam^t\|_F^2\leq C_0&\left(\|\bW^{t+1}-\bW^t\|_F^2\right.\\
&~~\left.+\|\bB^{t+1}-\bB^t\|_F^2+\|\boldsymbol{\mu}^{t+1}-\boldsymbol{\mu}^t\|_F^2\right)
\end{aligned}$$}
and 
{\small$$
\begin{aligned}\|\bmu^{t+1}-\bmu^t\|_F^2\leq& C_1\left(\|\bW^{t+1}-\bW^t\|_F^2+\|\bU^{t+1}-\bU^t\|_F^2\right)\\
&\hspace{-0.05cm}+C_1\rho_{\mu}^2\left(\|\bU^{t+1}\hspace{-0.05cm}-\hspace{-0.05cm}\bU^t\|_F^2+\|\bU^t\hspace{-0.05cm}-\hspace{-0.05cm}\bU^{t-1}\|_F^2\right).
\end{aligned}$$}
\end{lemma}
\begin{proof}
First, by combining the optimality conditions of (29b) and (29d) with (29e) and (29f), we get \vspace{-0.2cm}

{\small$$
\begin{aligned}
\blam^{t+1}&=\left(\mathbf{I}+\mathrm{i}Z_0\bB^{t+1}\right)^{-1}\bG\bW^{t+1}(\bmu^{t+1})^{\dagger}
\end{aligned}$$}
and
{\small$$\begin{aligned}
\bmu^{t+1}&=2((\bU^t)^{\dagger}\bG\bG^{\dagger}\bU^t)^{-1}(\bU^t)^{\dagger}\bG\bW^{t+1}\\
&\hspace{3cm}-\rho_{\mu}(\bU^{t+1}-\bU^t)^{\dagger}\bG\bW^{t+1}.
\end{aligned}$$}
\hspace{-0.1cm}Then, applying a similar argument as in Lemma \ref{lambdabound} gives the first assertion of Lemma 3. For $\bmu$,  using the expression above, we have\vspace{-0.2cm}

{\small\begin{equation}\label{mu}
\begin{aligned}
\bmu^{t+1}-\bmu^t&=2((\bU^t)^{\dagger}\bG\bG^{\dagger}\bU^t)^{-1}(\bU^t)^{\dagger}\bG\bW^{t+1}\\
&~~~\,-2((\bU^{t-1})^{\dagger}\bG\bG^{\dagger}\bU^{t-1})^{-1}(\bU^{t-1})^{\dagger}\bG\bW^{t}\\
&~~~-\rho_\mu(\bU^{t+1}-\bU^t)^{\dagger}\bG\bW^{t+1}\\
&~~~+\rho_\mu(\bU^{t}-\bU^{t-1})^{\dagger}\bG\bW^t.
\end{aligned}
\end{equation}}Note that for two full-rank matrices $\bX_1$ and $\bX_2$, the following inequality holds: \vspace{-0.2cm}

{\small$$
\begin{aligned}
&\|(\bX_1\bX_1^\dagger)^{-1}\bX_1-(\bX_2\bX_2^\dagger)^{-1}\bX_2\|_F\\
\leq &\|(\bX_1\bX_1^\dagger)^{-1}-(\bX_2\bX_2^\dagger)^{-1}\|_F\|\bX_1\|_2\\
&+\|(\bX_2\bX_2^\dagger)^{-1}\|_2\|\bX_1-\bX_2\|_F\\
\leq&\|(\bX_1\bX_1^\dagger)^{-1}\|_2\|(\bX_2\bX_2^\dagger)^{-1}\|_2\|\bX_1\|_2\|\bX_1\bX_1^\dagger-\bX_2\bX_2^\dagger\|_F\\
&+\|(\bX_2\bX_2^\dagger)^{-1}\|_2\|\bX_1-\bX_2\|_F\\
\leq&\left(\|(\bX_1\bX_1^\dagger)^{-1}\|_2\|(\bX_2\bX_2^\dagger)^{-1}\|_2\|\bX_1\|_2(\|\bX_1\|_2+\|\bX_2\|_2)\right.\\
&\left.~~+\|(\bX_2\bX_2^\dagger)^{-1}\|_2\right)\|\bX_1-\bX_2\|_F.
\end{aligned}$$}Applying the above inequality with $\bX_1=(\bU^{t-1})^{\dagger}\bG$ and $\bX_2=(\bU^{t})^{\dagger}\bG$, and utilizing 
\eqref{mu}, we can conclude that under the assumptions in Theorem 2, there exists $C_1>0$ such the second assertion in Lemma 3 holds.   
\end{proof}
 \begin{lemma}\label{lemma4}
 Define\vspace{-0.2cm}
 
  {\small$$\tilde{\mathcal{L}}_{{\boldsymbol{\rho}}}^{t}:=\mathcal{L}_{\boldsymbol{\rho}}^t+(\frac{C_0C_1\rho_\mu^2}{\rho_\lambda}+C_1\rho_\mu)\|\bU^t-{\bU}^{t-1}\|_F^2,$$}
  where $\mathcal{L}_{\boldsymbol{\rho}}^t:=\mathcal{L}_{\boldsymbol{\rho}}({\color{black}\bV^t},\bW^t,\bB^t,\bU^t,\blam^t,\bmu^t)$. It holds that \vspace{-0.2cm}
  
 {\small \begin{equation}\label{tildeL}
  \begin{aligned}
 & \tilde{\mathcal{L}}_{\boldsymbol{\rho}}^{t+1}\hspace{-0.05cm}-\hspace{-0.05cm}\tilde{\mathcal{L}}_{\boldsymbol{\rho}}^t\\
  \leq &-\frac{\rho_\mu}{2}\|{\color{black}\bV^{t+1}}-{\color{black}\bV^t}\|_F^2-(\frac{\xi}{4}-\frac{C_0}{\rho_\lambda})\|\bB^{t+1}-\bB^t\|_F^2\\
&-(1-\frac{C_0+C_0C_1}{\rho_\lambda}-\frac{C_1}{\rho_\mu})\|\bW^{t+1}-\bW^t\|_F^2\\
&-(\frac{\rho_\lambda}{2}-\frac{C_0C_1}{\rho_\lambda}\hspace{-0.05cm}-\hspace{-0.05cm}\frac{C_1}{\rho_\mu}\hspace{-0.05cm}-\hspace{-0.05cm}\frac{2C_0C_1\rho_\mu^2}{\rho_\lambda}\hspace{-0.05cm}-\hspace{-0.05cm}2C_1\rho_\mu)\|\bU^{t+1}\hspace{-0.05cm}-\hspace{-0.05cm}\bU^t\|_F^2.\vspace{-0.1cm}
  \end{aligned}
  \end{equation}}
 \end{lemma}
\begin{proof}
We first note that the objective functions of the {\color{red}$\bV$}-, $\bW$-, $\bB$-, and $\bU$-subproblems in (29a)\,--\,(29d) are all strongly convex, with modulus $\rho_\mu$, $2$, $\frac{\xi}{2}$, and $\rho_\lambda$, respectively.  Hence, similar to the proof of Lemma \ref{ALbound}, we have \vspace{-0.2cm}
 
 {\small$$
 \begin{aligned}
& \mathcal{L}_{\boldsymbol{\rho}}^{t+1}-\mathcal{L}_{\boldsymbol{\rho}}^t\\
 \leq&\hspace{-0.05cm} -\frac{\rho_{\mu}}{2}\|{\color{black}\bV^{t+1}}\hspace{-0.05cm}-\hspace{-0.05cm}{\color{black}\bV^t}\|_F^2-\frac{\xi}{4}\|\bB^{t+1}\hspace{-0.05cm}-\hspace{-0.05cm}\bB^t\|_F^2-\hspace{-0.05cm}\|\bW^{t+1}-\bW^t\|_F^2\\
 &-\frac{\rho_\lambda}{2}\|\bU^{t+1}-\bU^t\|_F^2+\frac{\|\blam^{t+1}-\blam^t\|_F^2}{\rho_{\lambda}} + \frac{\|\bmu^{t+1}-\bmu^t\|_F^2}{\rho_{\mu}}\\
\leq &-\frac{\rho_\mu}{2}\|{\color{black}\bV^{t+1}}-{\color{black}\bV^t}\|_F^2-(\frac{\xi}{4}-\frac{C_0}{\rho_\lambda})\|\bB^{t+1}-\bB^t\|_F^2\\
&-(1-\frac{C_0+C_0C_1}{\rho_\lambda}-\frac{C_1}{\rho_\mu})\|\bW^{t+1}-\bW^t\|_F^2\\
&-(\frac{\rho_\lambda}{2}-\frac{C_0C_1}{\rho_\lambda}\hspace{-0.05cm}-\hspace{-0.05cm}\frac{C_1}{\rho_\mu}-\frac{C_0C_1\rho_\mu^2}{\rho_\lambda}-C_1\rho_\mu)\|\bU^{t+1}-\bU^t\|_F^2\\
&+(\frac{C_0C_1\rho_\mu^2}{\rho_\lambda}+C_1\rho_\mu)\|\bU^t-{\bU}^{t-1}\|_F^2,
 \end{aligned}$$}
 \hspace{-0.35cm} where the first inequality follows from the strong convexity of each subproblem and the update rules of the Lagrange multipliers, and the second inequality is obtained by applying Lemma 3. Then Lemma 4 follows immediately by adding $({C_0C_1\rho_\mu^2}/{\rho_\lambda}+C_1\rho_\mu)(\|\bU^{t+1}-{\bU}^{t}\|_F^2-\|\bU^t-{\bU}^{t-1}\|_F^2)$ to both sides of the above expression.

\end{proof}
\emph{Proof of Theorem 2: }It is straightforward to see that we can choose sufficiently large $c_0$ and $\rho_0$ in Theorem 2 such that each coefficient in the r.h.s. of Lemma 4 is positive.  Then, by summing over \eqref{tildeL} from $t=0$ to $t=T$ and noting the lower boundedness of $\tilde{L}_{\rho}^t$, we can show that \vspace{-0.2cm}

{\small$$
\begin{aligned}
&\|{\color{black}\bV^{t+1}-\bV^t}\|_F\to 0, ~\|\bW^{t+1}-\bW^t\|_F\to 0,\\
&\|\bB^{t+1}-\bB^t\|_F\to 0,~\|\bU^{t+1}-\bU^{t}\|_F\to 0.
\end{aligned}$$}
\hspace{-0.15cm}Following similar steps as in Appendix \ref{app:AB}, we can get the desired result. 

{\section{Solution to (40)}
In this appendix, we derive the optimal solution to problem (40).  For notational simplicity, let $v:=v^t$. According to the KKT condition, its solution can be expressed as 
$$\mathbf{w}_k^*=\left(\bQ+\alpha^*\mathbf{I}_N\right)^{-1}\left(\mathbf{c}_k+\frac{\tau}{2v}\bw_k^t\right),~k=1,2,\dots, K,$$
where 
\begin{equation}\label{alpha:eta}
\alpha^*=v\mu_0(\mu_0\|\bW^*\|_F^2+P_c)+\frac{\tau}{2v}+\eta^*,
\end{equation} and $\eta^*$ is the Lagrangian multiplier associated with the transmit power constraint $\|\bW\|_F^2\leq P_T$.  
Let 
\begin{equation}\label{eq:Walpha}
\begin{aligned}
\bW(\alpha)=\left(\bQ+\alpha\mathbf{I}_N\right)^{-1}&\left(\mathbf{C}+\frac{\tau}{2v}\bW^t\right),
\end{aligned}
\end{equation}
where $\mathbf{C}=[\mathbf{c}_1,\mathbf{c}_2,\dots, \mathbf{c}_K].$ The Lagrangian multiplier $\eta^*$ is determined according to the complementary slackness condition:
\begin{equation}\label{slackness}
\big(\|\bW(\alpha)\|_F^2-P_T\big)\eta=0,
\end{equation}
together with the feasibility conditions  
$$\eta\geq 0\text{ and }\|\bW(\alpha)\|_F^2\leq P_T.$$
As discussed below (18), $\bW(\alpha)$ defined in \eqref{eq:Walpha} can be expressed as 
$$\|\bW(\alpha)\|_F^2=\sum_{n=1}^N\frac{\|\boldsymbol{\phi}_n\|^2}{(d_n+\alpha)^2},$$
where $\boldsymbol{\phi}_n^T$ is the $n$-th row of the matrix 
$\bU_Q^\dagger(\bC+\tau\bW^t/(2v))$, and $\bQ=\bU_Q\bD_Q\bU_Q^\dagger$ is the eigenvalue decomposition of $\bQ$ with $\bD_Q=\text{diag}(d_1,d_2,\dots,d_N)$. Using this, and combining \eqref{alpha:eta} and \eqref{slackness} together with the feasibility conditions, we obtain that $(\eta^*,\alpha^*)$ is the solution to the following equations:
\begin{subequations}\label{solution}
\begin{align}
&\alpha=v\mu_0^2\sum_{n=1}^N\frac{\|\boldsymbol{\phi}_n\|^2}{(d_n+\alpha)^2}+v\mu_0P_c+\frac{\tau}{2v}+\eta,\label{condition:a}\\
&\left(\sum_{n=1}^N\frac{\|\boldsymbol{\phi}_n\|^2}{(d_n+\alpha)^2}-P_T\right)\eta=0,\label{condition:b}\\
&\sum_{n=1}^N\frac{\|\boldsymbol{\phi}_n\|^2}{(d_n+\alpha)^2}\leq P_T,~\eta\geq 0.
\end{align}
\end{subequations}
We next  solve the above equations to obtain $(\eta^*,\alpha^*)$.
 
\emph{1) Case 1: $\eta=0$}. We first consider the case where the transmit power constraint is inactive. Define
\begin{equation}\label{def:Falpha}
F(\alpha):=\alpha-v\mu_0^2\sum_{n=1}^N\frac{\|\boldsymbol{\phi}_n\|^2}{(d_n+\alpha)^2}-v\mu_0P_c-\frac{\tau}{2v}.
\end{equation}
With $\eta=0$, \eqref{condition:a} and \eqref{condition:b} reduce to $F(\alpha)=0$. Note that $F(\alpha)$ is monotonically increasing in $\alpha\geq 0$ with $F(0)<0$ and $\lim_{\alpha\to\infty}F(\alpha)=\infty$. Therefore, there exists a unique positive solution to $F(\alpha)=0$, denoted by $\alpha_0> 0$, which can be efficiently solved by one-dimensional bisection. If 
\begin{equation}\label{Walpha0}
\|\bW(\alpha_0)\|_F^2\leq P_T,
\end{equation} we conclude that $$\eta^*=0 \text{ and } \alpha^*=\alpha_0.$$ Otherwise, we turn to Case 2.

\emph{2)  Case 2: $\eta>0$.}  In this case, the transmit power constraint is active, and  \eqref{condition:a} and \eqref{condition:b} reduce to 
\begin{subequations}
\begin{align}
&\alpha=v\mu_0^2P_T+v\mu_0P_c+\frac{\tau}{2v}+\eta,\label{condition:a2}\\
&\sum_{n=1}^N\frac{\|\boldsymbol{\phi}_n\|^2}{(d_n+\alpha)^2}=P_T.\label{condition:b2}
\end{align}\end{subequations}
Since \eqref{Walpha0} fails, i.e.,  $\|\bW(\alpha_0)\|_F^2>P_T$, the solution to \eqref{condition:b2} satisfies $\alpha^*> \alpha_0\geq 0$, which can be efficiently obtained via a one-dimensional bisection search.  Then from \eqref{condition:a2}, 
$$\eta^*=\alpha^*- v\mu_0^2P_T-v\mu_0P_c-\frac{\tau}{2v}.$$ 
To show that $(\eta^*,\alpha^*)$ is the solution to \eqref{solution}, it remains to prove that $\eta^*> 0$. Using \eqref{condition:b2} and the definition in \eqref{def:Falpha}, $\eta^*$ can be expressed as $\eta^*=F(\alpha^*)$. Since $\alpha^*>\alpha_0$ and $F(\alpha)$ is monotonically increasing, we can conclude that 
$$\eta^*=F(\alpha^*)> F(\alpha_0)=0.$$
Therefore, $(\eta^*,\alpha^*)$ is the solution to \eqref{solution}.}
\section{Convergence Behavior of the Proposed Algorithms}
\setcounter{figure}{7}
{\color{black}\begin{figure}
\subfigure[{\color{black}Objective value (sum-rate).}]{\includegraphics[width=0.23\textwidth]{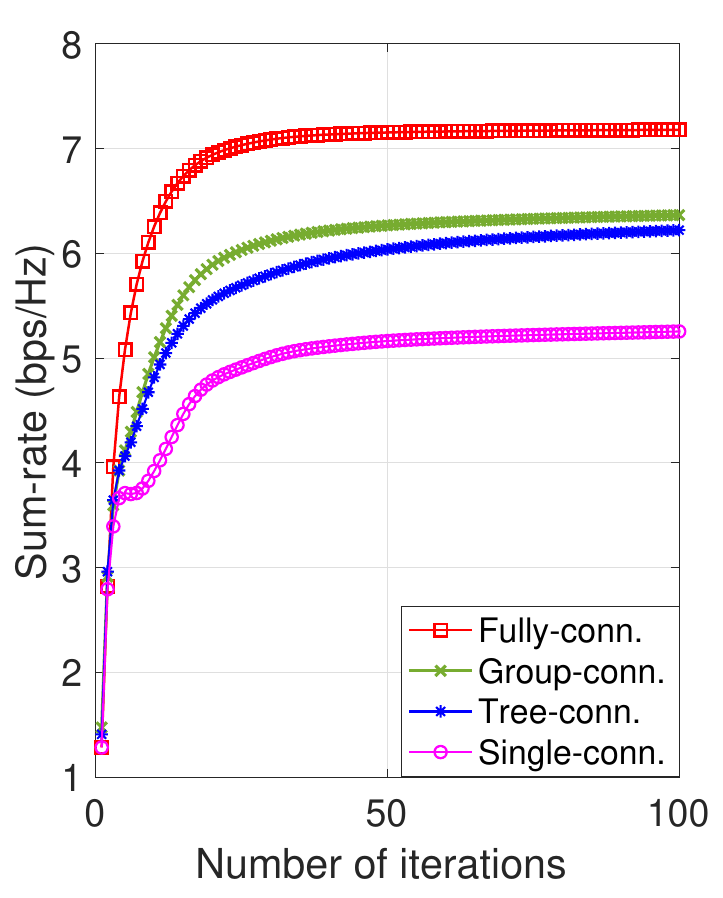}}
\subfigure[{\color{black}Constraint violation.}]{\includegraphics[width=0.23\textwidth]{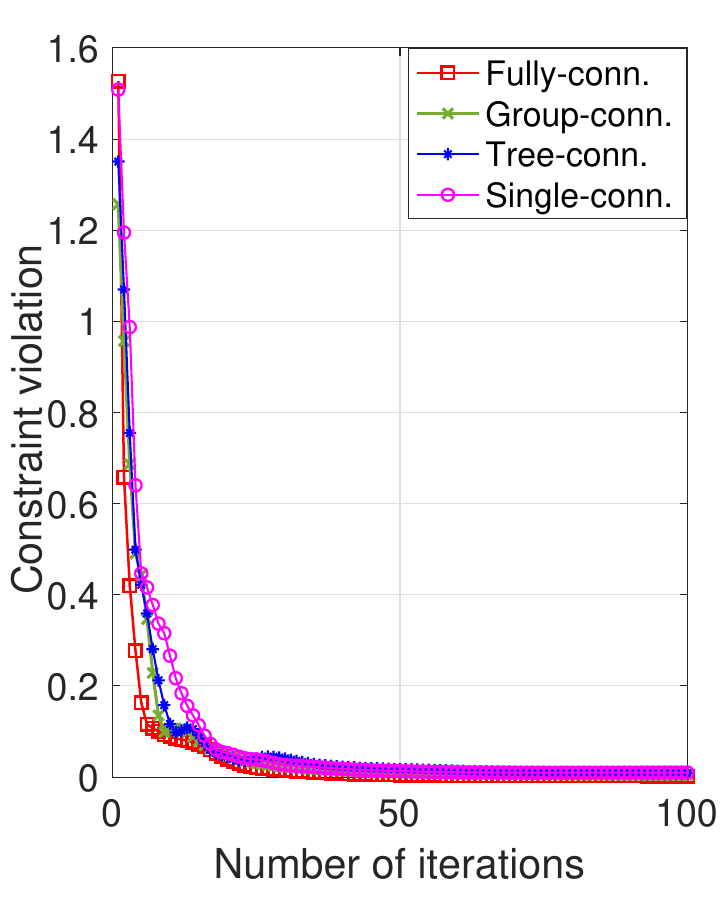}}
\centering
\caption{{\color{black}Convergence behavior of the pp-ADMM algorithm for the sum-rate maximization problem ($N=K=4, M=32, G_s=4$).}  }
\label{fig:converge}
\end{figure}

\begin{figure}
\subfigure[{\color{black}Objective value (transmit power).}]{\includegraphics[width=0.23\textwidth]{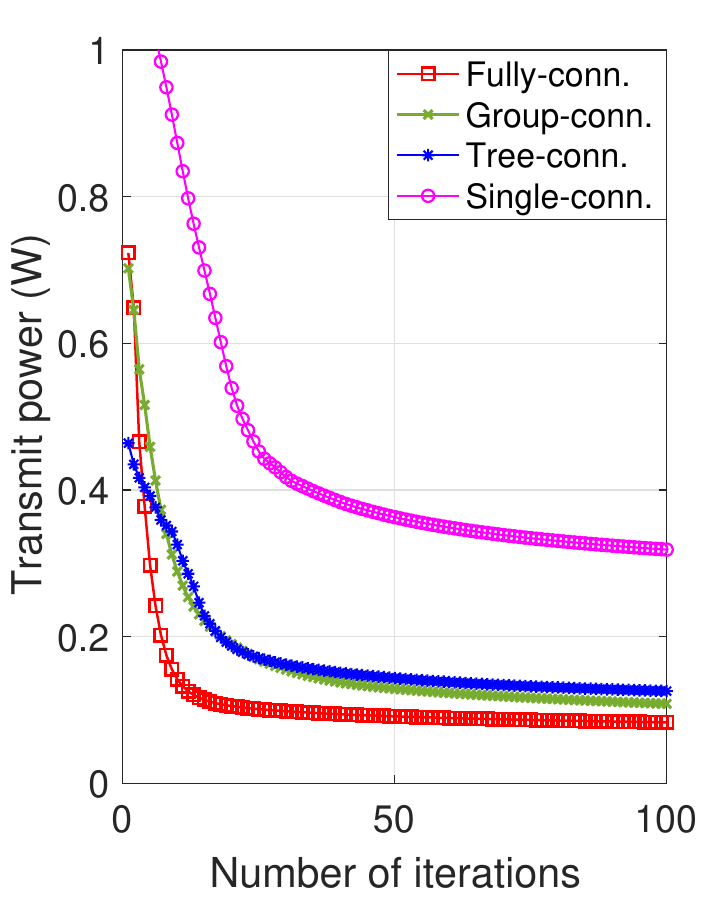}}
\subfigure[{\color{black}Constraint violation.}]{\includegraphics[width=0.23\textwidth]{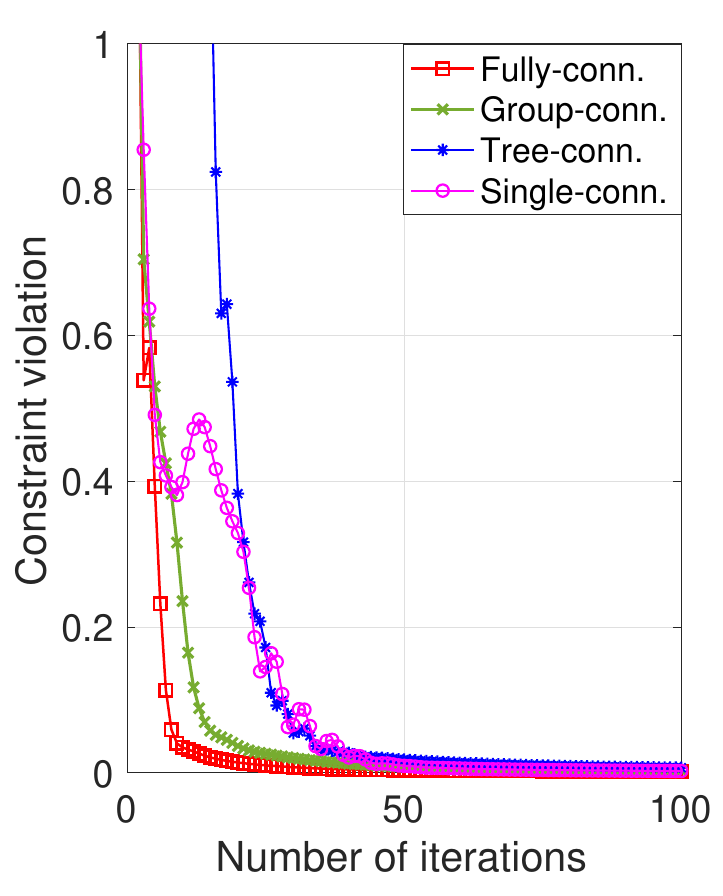}}
\centering
\caption{{\color{black}Convergence behavior of the pp-ADMM algorithm for the transmit power minimization problem  ($N=K=4, M=32, G_s=4, \Gamma_k=10$ dB). }}
\label{fig:converge2}
\end{figure}
In this appendix, we numerically verify the convergence of the proposed pp-ADMM algorithms for the sum-rate maximization problem and the transmit power minimization problem, respectively. For these two problems, the constraint violations are measured by $\|(\mathbf{I}-\mathrm{i}Z_0\bB)\bU-(\mathbf{I}+\mathrm{i}Z_0\bB)\bH\|_F$ and $\|\bV-\bU^\dagger\bG\bW\|_F+\|(\mathbf{I}-\mathrm{i}Z_0\bB)\bU-(\mathbf{I}+\mathrm{i}Z_0\bB)\bH\|_F$, respectively. As shown in the figures, for all considered architectures, the proposed pp-ADMM algorithm converges rapidly. The objective values (i.e., sum-rate and transmit power) converge and the constraint violations vanish within $100$ iterations,  ensuring the computational efficiency of the proposed algorithm.}

\vspace{-0.05cm}
  \bibliographystyle{IEEEtran}
\bibliography{BDRIS}

\end{document}